\pdfoutput=1

\def\oakland{0}
\def\eurocrypt{0}
\def\makeappendix{1}

\ifnum\oakland=0
	\ifnum\eurocrypt=0
      \documentclass[letter]{article}
      \def\comments{0}
      \def\todo{0}
      \def\toc{1}
	\else
    	\documentclass[runningheads]{llncs}
        \def\comments{0}
		\def\todo{0}
		\def\toc{0}
    \fi
\else
	\documentclass[conference]{IEEEtran}
	
	\def\comments{0}
	\def\todo{0}
	\def\toc{0}
\fi

\usepackage{preamble}
\usepackage{multirow}

\newcommand{\bits}{\{0,1\}}
\newcommand{\univx}{\mathcal{X}}
\newcommand{\univy}{\mathcal{Y}}

\newcommand{\prob}{\mathbb{P}}

\newcommand{\half}{\frac{1}{2}}

\newcommand{\shuff}{S}

\newenvironment{mymath}{\ifnum\oakland=1$\else$$\fi}{\ifnum\oakland=1$\else$$\fi}
  
\newcommand\blfootnote[1]{%
  \begingroup
  \renewcommand\thefootnote{}\footnote{#1}%
  \addtocounter{footnote}{-1}%
  \endgroup
}

\ifnum\eurocrypt=0
	\title{Distributed Differential Privacy via Shuffling%
	\thanks{\textcopyright~ IACR 2019. This article is a minor revision of the version published by Springer-Verlag available at 10.1007/978-3-030-17653-2\_13. A table of contents and appendices are included. Some of these results are based on previous, unpublished work by two of the authors ~\cite{ZhilyaevZ17}.}}
\else
	\title{Distributed Differential Privacy via Shuffling}
\fi

\ifnum\oakland=0
	\ifnum\eurocrypt=0
  \author{
  \makebox[1.5in]{\hfill Albert Cheu\thanks{Khoury College of Computer Sciences, Northeastern University. Research supported by NSF award CCF-1718088.  \href{mailto:cheu.a@husky.neu.edu}{\texttt{cheu.a@husky.neu.edu}}} \hfill} \and 
  \makebox[1.5in]{\hfill Adam Smith\thanks{Computer Science Department, Boston University.   Research supported by NSF awards IIS-1447700 and AF-1763786 and a Sloan Foundation Research Award.  \href{mailto:ads22@bu.edu}{\texttt{ads22@bu.edu}}.
  } \hfill} \and 
  \makebox[1.5in]{\hfill Jonathan Ullman\thanks{Khoury College of Computer Sciences, Northeastern University.  Research supported by NSF awards CCF-1718088, CCF-1750640, CNS-1816028 and a Google Faculty Research Award. \href{mailto:jullman@ccs.neu.edu}{\texttt{jullman@ccs.neu.edu}}} \hfill} \and 
  \makebox[1.5in]{\hfill David Zeber\thanks{Mozilla Foundation.  \href{mailto:dzeber@mozilla.com}{\texttt{dzeber@mozilla.com}}} \hfill} \and 
  \makebox[1.5in]{\hfill Maxim Zhilyaev\thanks{\href{mailto:maxim.zhilyaev@gmail.com}{\texttt{maxim.zhilyaev@gmail.com}}} \hfill}
  }
  
  \else
 	\author{
    Albert Cheu \inst{1}\textsuperscript{(\Letter)} \and Adam Smith\inst{2} \and Jonathan Ullman\inst{1} \and \\ David Zeber\inst{3} \and Maxim Zhilyaev\inst{4}
    }
    
    \institute{
    Khoury College of Computer Sciences, Northeastern University \email{cheu.a@husky.neu.edu, jullman@ccs.neu.edu}
    \and Computer Science Department, Boston University \email{ads22@bu.edu}
    \and Mozilla Foundation \email{dzeber@mozilla.com}
    \and \email{maxim.zhilyaev@gmail.com}
    }
    \authorrunning{A. Cheu et al.}
  \fi
\fi

\begin{document}

\ifnum\todo=1
    {\color{DarkBlue} \begin{center}
    {\Large To Do:}
    \end{center}
    \begin{itemize}
    
    \item Check funding information.
    
    \item Cite Maxim and David's work on github.
    
    \item Could still add some discussion of the histograms lower bound to the intro, but I opted against giving any detailed discussion.
    
    \end{itemize}}
    
    \vfill\newpage
\fi

\maketitle

\begin{abstract}
  We consider the problem of designing scalable, robust protocols for computing statistics about sensitive data. Specifically, we look at how best to design differentially private protocols in a distributed setting, where each user holds a private datum. The literature has mostly considered two models:  the ``central'' model, in which a trusted server collects users' data in the clear, which allows greater accuracy; and the ``local'' model, in which users individually randomize their data, and need not trust the server, but accuracy is limited.  Attempts to achieve the accuracy of the central model without a trusted server have so far focused on variants of cryptographic \ifnum\eurocrypt=1 multiparty computation (MPC)\else secure function evaluation\fi, which limits scalability.
 
\ifnum\eurocrypt=1 \vspace{5pt} \fi
  In this paper, we initiate the analytic study of a \emph{shuffled model} for distributed differentially private algorithms, which lies between the local and central models. This simple-to-implement model, a special case of the ESA framework of \ifnum\eurocrypt=0 \citet{Bittau+17}\else \cite{Bittau+17}\fi, augments the local model with an anonymous channel that randomly permutes a set of user-supplied messages. 
  For sum queries, we show that this model provides the power of the central model while avoiding the need to trust a central server and the complexity of cryptographic secure function evaluation. More generally, we give evidence that the power of the shuffled model lies strictly between those of the central and local models: for a natural restriction of the model, we show that shuffled protocols for a widely studied \emph{selection} problem require exponentially higher sample complexity than do central-model protocols.
\end{abstract}


\ifnum\oakland=0
    \ifnum\eurocrypt=0
        \vfill
        \newpage
    \fi
\fi

\ifnum\toc=1
\tableofcontents
\vfill
\newpage

\fi

\setcounter{footnote}{0}

\section{Introduction}

\ifnum\eurocrypt=1\blfootnote{The full version of this paper is accessible \href{https://arxiv.org/abs/1808.01394}{on arXiv}}\fi
The past few years has seen a wave of commercially deployed systems~\citep{ErlingssonPK14,AppleA} for analysis of users' sensitive data in the \emph{local model of differential privacy (LDP)}.  LDP systems have several features that make them attractive in practice, and limit the barriers to adoption.  Each user only sends private data to the data collector, so users do not need to fully trust the collector, and the collector is not saddled with legal or ethical obligations.  Moreover, these protocols are relatively simple and scalable, typically requiring each party to asynchronously send just a single short message.

However, the local model imposes strong constraints on the utility of the algorithm.  These constraints preclude the most useful differentially private algorithms, which require a \emph{central model} where the users' data is sent in the clear, and the data collector is trusted to perform only differentially private computations.  Compared to the central model, the local model requires enormous amounts of data, both in theory and in practice (see e.g.~\cite{KasiviswanathanLNRS08} and the discussion in~\cite{Bittau+17}).  Unsurprisingly, the local model has so far only been used by large corporations like Apple and Google with billions of users.

In principle, there is no dilemma between the central and local models, as any algorithm can be implemented without a trusted data collector using cryptographic \emph{multiparty computation (MPC)}.  However, despite dramatic recent progress in the area of practical MPC, existing techniques still require large costs in terms of computation, communication, and number of rounds of interaction between the users and data collector, and are considerably more difficult for companies to extend and maintain.

In this work, we initiate the analytic study of an intermediate model for distributed differential privacy called the \emph{shuffled model}.  This model, a special case of the ESA framework of \cite{Bittau+17}, augments the standard model of local differential privacy with an anonymous channel (also called a shuffler) that collects messages from the users, randomly permutes them, and then forwards them to the data collector for analysis.  For certain applications, this model overcomes the limitations on accuracy of local algorithms while preserving many of their desirable features.  However, under natural constraints, this model is dramatically weaker than the central model.  In more detail, we make two primary contributions:

\begin{itemize}
\item
  We give a simple, non-interactive algorithm in the shuffled model for estimating a single Boolean-valued statistical query (also known as a counting query) that essentially matches the error achievable by centralized algorithms.  We also show how to extend this algorithm to estimate a bounded real-valued statistical query, albeit at an additional cost in communication.  These protocols are sufficient to implement any algorithm in the \emph{statistical queries model}~\cite{Kearns93}, which includes methods such as gradient descent.

\item
 We consider the ubiquitous \emph{variable-selection problem}---a simple but canonical optimization problem.  Given a set of counting queries, the variable-selection problem is to identify the query with nearly largest value (i.e.~an ``approximate argmax'').  We prove that the sample complexity of variable selection in a natural restriction of the shuffled model is exponentially larger than in the central model.  The restriction is that each user send only a single message into the shuffle, as opposed to a set of messages, which we call this the \emph{one-message} shuffled model.  Our positive results show that the sample complexity in the shuffled model is polynomially smaller than in the local model. Taken together, our results give evidence that the central, shuffled, and local models are strictly ordered in the accuracy they can achieve for selection.  Our lower bounds follow from a structural result showing that any algorithm that is private in the one-message shuffled model is also private in the local model with weak, but non-trivial, parameters.
\end{itemize}

In concurrent and independent work, Erlingsson et al.~\cite{ErlingssonFMRTT19} give conceptually similar positive results for local protocols aided by a shuffler.  We give a more detailed comparison between our work and theirs after giving a thorough description of the model and our results (Section~\ref{sec:comparison})

\subsection{Background and Related Work}

\myparagraph{Models for Differentially Private Algorithms.}
Differential privacy \citep{DworkMNS06} is a restriction on the algorithm that processes a dataset to provide statistical summaries or other output. It ensures that, no matter what an attacker learns by interacting with the algorithm, it would have learned nearly the same thing whether or not the dataset contained any particular individual's data~\citep{KasiviswanathanS08}.
Differential privacy is now widely studied, and algorithms satisfying the criterion are increasingly deployed \citep{Abowd18,apple,ErlingssonPK14}.

There are two well-studied models for implementing differentially-private algorithms. In the \emph{central model}, raw data are collected at a central server where they are processed by a differentially private algorithm. In the \emph{local model} \citep{Warner65,EGS03,DworkMNS06}, each individual applies a differentially private algorithm locally to their data and shares only the output of the algorithm---called a report or response---with a server that aggregates users' reports. The local model allows individuals to retain control of their data since privacy guarantees are enforced directly by their devices. It avoids the need for a single, widely-trusted entity and the resulting single point of security failure.  The local model has witnessed an explosion of research in recent years, ranging from theoretical work to deployed implementations.  A complete survey is beyond the scope of this paper.

%
%
Unfortunately, for most tasks there is a large, unavoidable gap between the accuracy that is achievable in the two models. \ifnum\eurocrypt=0 \citet{BeimelNO08} and \citet{ChanSS12} \else \cite{BeimelNO08} and \cite{ChanSS12} \fi  show that estimating the sum of bits, one held by each player, requires error $\Omega(\sqrt{n}/\eps)$ in the local model, while an error of just $O(1/\eps)$ is possible the central model. \cite{DuchiJW13} extended this lower bound to a wide range of natural problems, showing that the error must blowup by at least $\Omega(\sqrt{n})$, and often by an additional factor growing with the data dimension.
More abstractly, \ifnum\eurocrypt=0 \citet{KasiviswanathanLNRS08} \else \cite{KasiviswanathanLNRS08} \fi showed that the power of the local model is equivalent to the \emph{statistical query model} \citep{Kearns93} from learning theory. They used this to show an exponential separation between the accuracy and sample complexity of local and central algorithms.  Subsequently, an even more natural separation arose for the variable-selection problem~\citep{DuchiJW13,selection}, which we also consider in this work.

\myparagraph{Implementing Central-Model Algorithms in Distributed Models.}
In principle, one could also use the powerful, general tools of modern cryptography, such as multiparty computation (MPC), or secure function evaluation, to simulate central model algorithms in a setting without a trusted server \citep{DworkKMMN06}, but such algorithms currently impose bandwidth and liveness constraints that make them impractical for large deployments.
In contrast, Google \citep{ErlingssonPK14} now uses local differentially private protocols to collect certain usage statistics from hundreds of millions of users' devices.


A number of specific, efficient MPC algorithms have  been proposed for differentially private functionalities. They generally either (1) focus on simple summations and require a single ``semi-honest''/``honest-but-curious'' server that aggregates user answers, as in \ifnum\eurocrypt=0 \citet{ShiCRCS11,ChanSS12-aggregation,Bonawitz+17} \else \cite{ShiCRCS11,ChanSS12-aggregation,Bonawitz+17} \fi; or (2) allow general computations, but require a network of servers, a majority of whom are assumed to behave honestly, as in \ifnum\eurocrypt=0 \citet{CorriganB17-prio}\else \cite{CorriganB17-prio}\fi. As they currently stand, these approaches have a number of drawbacks: they either require users to trust that a server maintained by a service provided is behaving (semi-)honestly, or they require that a coalition of service providers collaborate to run protocols that reveal to each other who their users are and \emph{what computations they are performing on their users' data}. It is possible to avoid these issues by combining anonymous communication layers and MPC protocols for universal circuits but, with current techniques, such modifications destroy the efficiency gains relative to generic MPC. 

Thus, a natural question---relevant no matter how the state of the art in MPC evolves---is to identify simple (and even minimal) primitives that can be implemented via MPC in a distributed model and are expressive enough to allow for sophisticated private data analysis.  In this paper, we show that shuffling is a powerful primitive for differentially private algorithms.


\myparagraph{Mixnets.}
One way to realize the shuffling functionality is via a mixnet.
A \emph{mix network}, or \emph{mixnet}, is a protocol involving several computers that takes as input a sequence of encrypted messages, and outputs a uniformly random permutation of those messages' plaintexts. Introduced by \cite{Chaum81-mixnets}, the basic idea now exists in many variations. In its simplest instantiation, the network consists of a sequence of servers, whose identities and ordering are public information.\footnote{Variations on this idea based on \emph{onion routing} allow the user to specify a secret path through a network of mixes.}
Messages, each one encrypted with all the servers' keys, are submitted by users to the first server. Once enough messages have been submitted, each server in turn performs a \emph{shuffle} in which the server removes one layer of encryption and sends a permutation of the messages to the next server. In a \emph{verifiable shuffle}, the server also produces a cryptographic proof that the shuffle preserved the multi-set of messages. The final server sends the messages to their final recipients, which might be different for each message. A variety of efficient implementations of mixnets with verifiable shuffles exist (see, e.g., \cite{KwonLDF16-riffle,Bittau+17} and citations therein).

\jnote{Old text in blue.  New text in red.}
Another line of work~\cite{HLZZ15,TGLZZ17} shows how to use differential privacy {\em in addition} to mixnets to make communication patterns differentially private for the purposes of anonymous computation.  Despite the superficial similarly, this line of work is orthogonal to ours, which is about how to use mixnets themselves to achieve (more accurate) differentially private data analysis.

\begin{figure}
\begin{framed}
\centering
\begin{subfigure}{.5\linewidth}
  \centering
  \includegraphics[width=.65\linewidth]{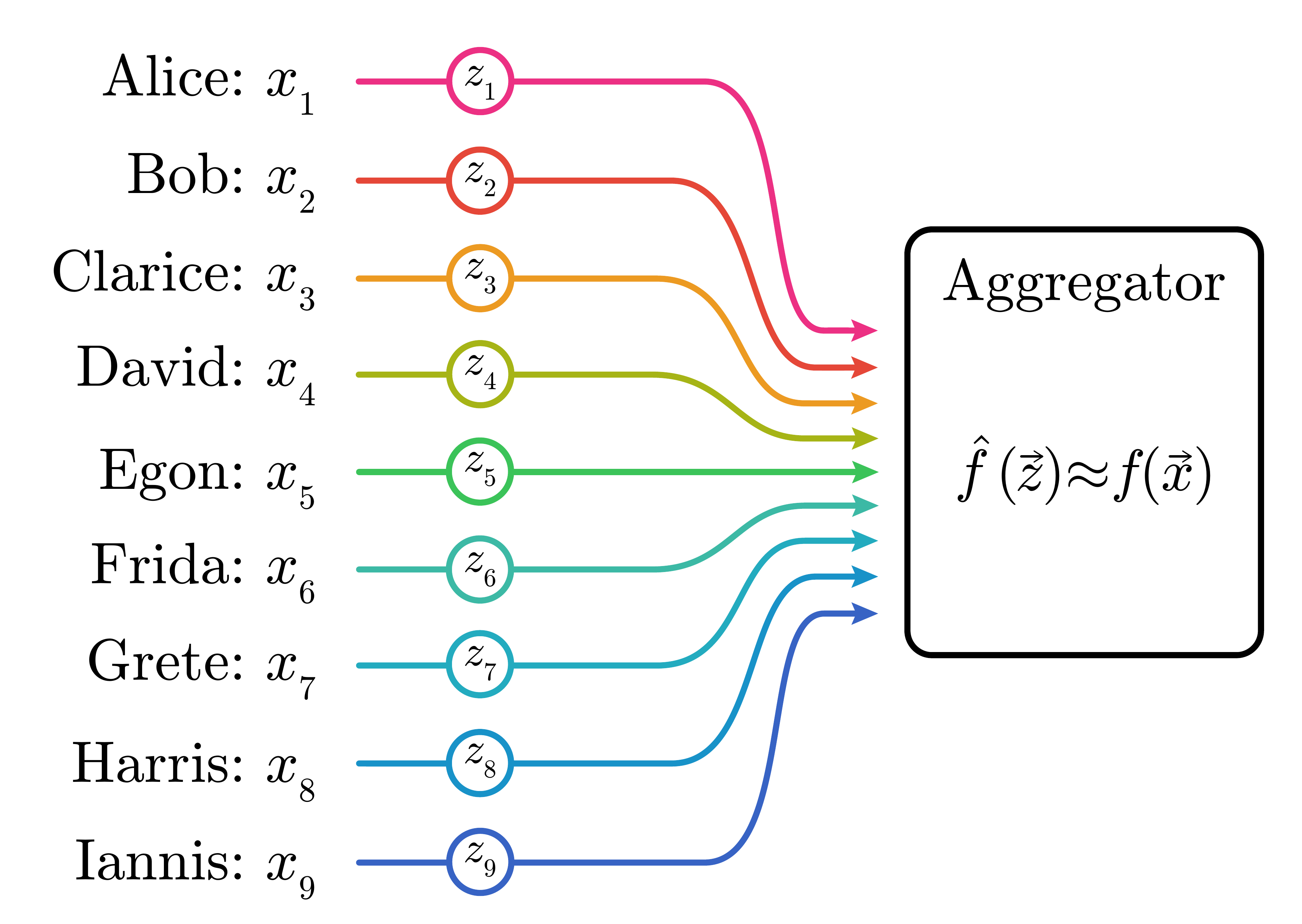}
  \label{fig:sub1}
\end{subfigure}%
\begin{subfigure}{.5\linewidth}
  \centering
  \includegraphics[width=.85\linewidth]{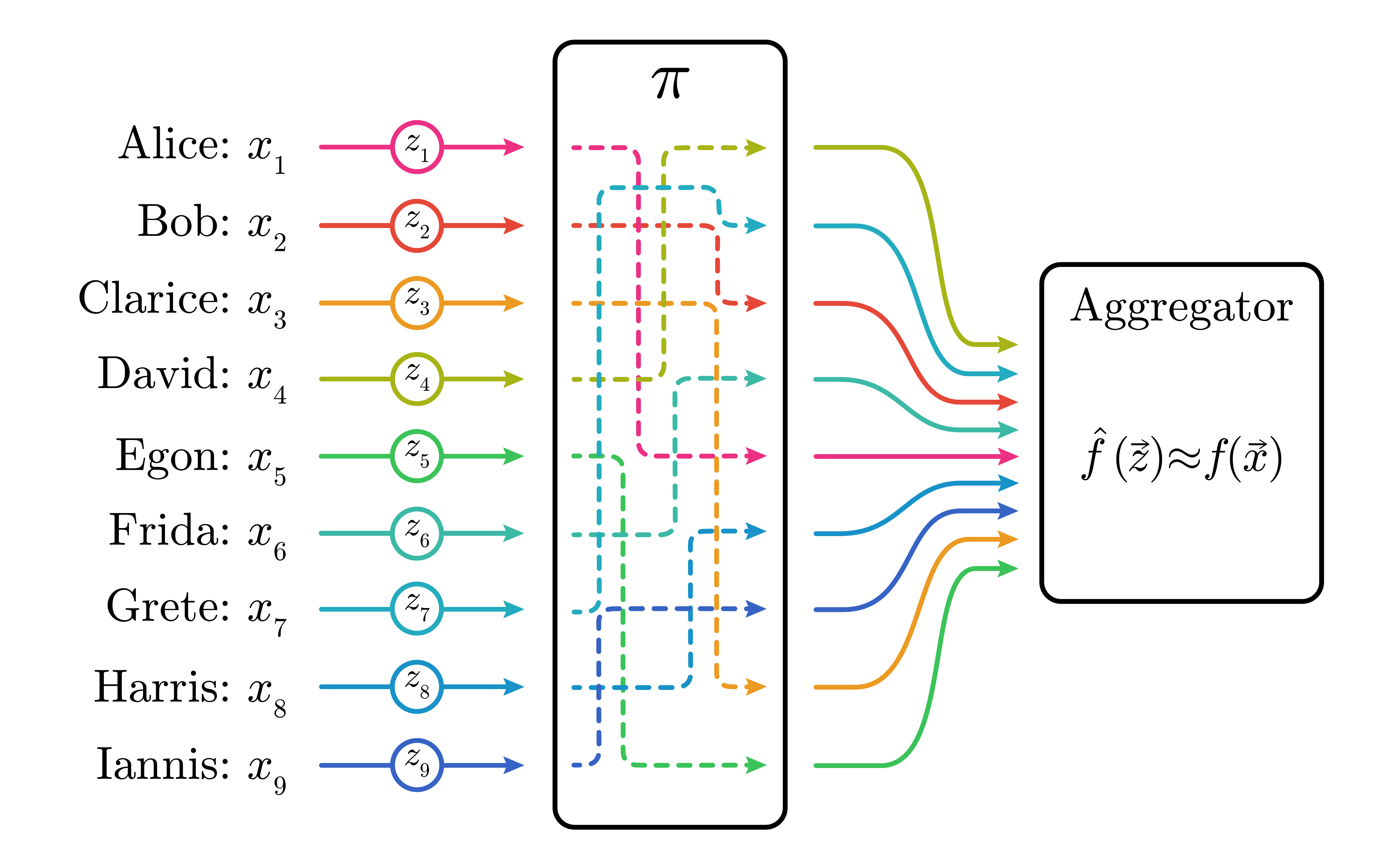}
  \label{fig:sub2}
\end{subfigure}
{\footnotesize Prototypical (one-message) protocols in the local model (left) and the shuffled model (right).}
\end{framed}
\end{figure}

\myparagraph{Shufflers as a Primitive for Private Data Analysis.}
This paper studies how to use a shuffler (e.g. a mixnet) as a cryptographic primitive to implement differentially-private algorithms. \ifnum\eurocrypt=0\citet{Bittau+17} \else Bittau et al \cite{Bittau+17} \fi propose a general framework, dubbed \emph{encode-shuffle-analyze} (or \emph{ESA}), which generalizes the local and central models by allowing a local randomized encoding step $E$ performed on user devices, a permutation step $S$ in which encrypted encodings are shuffled, and a final randomized process $A$ that analyzes the permuted encodings. We ask what privacy guarantee can be provided if we rely only on the local encoding $E$ and the shuffle $S$---the analyst $A$ is untrusted. In particular, we are interested in protocols that are substantially more accurate than is possible in the local model (in which the privacy guarantee relies entirely on the encoding $E$). This general question was left open by \ifnum\eurocrypt=0\citet{Bittau+17}\else\cite{Bittau+17} \fi. 

One may think of the shuffled model as specifying a highly restricted MPC primitive on which we hope to base privacy. 
Relative to general MPC, the use of mixnets for shuffling provides several advantages: First, there already exist a number of highly efficient implementations. Second, their trust model is simple and  robust---as long as a single one of the servers performs its shuffle honestly, the entire process is a uniformly random permutation, and our protocols' privacy guarantees will hold.   The architecture and trust guarantees are also easy to explain to nonexperts
(say, with metaphors of shuffled cards or shell games). Finally, mixnets automatically provide a number of additional features that are desirable for data collection:  they can maintain secrecy of a company's user base, since each company's users could use that company's server as their first hop; and they can maintain secrecy of the company's computations, since the specific computation is done by the  analyst.  
Note that we think of a mixnet here as operating on large batches of messages, whose size is denoted by $n$. (In implementation, this requires a fair amount of latency, as the collection point must receive sufficiently many messages before proceeding---see Bittau et al.~\cite{Bittau+17}).

Understanding the possibilities and limitations of shuffled protocols for private data analysis is interesting from both theoretical and practical perspectives. It provides an intermediate abstraction, and we give evidence that it lies strictly between the central and local models.  Thus, it sheds light on the minimal cryptographic primitives needed to get the central model's accuracy.  It also provides an attractive platform for near-term deployment~\cite{Bittau+17}, for the reasons listed above.

For the remainder of this paper, we treat the shuffler as an abstract service that randomly permutes a set of messages. We leave a discussion of the many engineering, social, and cryptographic implementation considerations to future work.



\section{Overview of Results}

\myparagraph{The Shuffled Model.}
In our model, there are $n$ users, each with data $x_i \in \cX$.  Each user applies some \emph{encoder} $R \from \cX \to \cY^m$ to their data and sends the \emph{messages} $(y_{i,1},\dots,y_{i,m}) = R(x_i)$.  In the \emph{one-message shuffled model}, each user sends $m=1$ message. The $n \cdot m$ messages $y_{i,j}$ are sent to a \emph{shuffler} $\shuff \from \cY^* \to \cY^*$ that takes these messages and outputs them in a uniformly random order.  The shuffled set of messages is then passed through some \emph{analyzer} $A \from \cY^* \to \cZ$ to estimate some function $f(x_1,\dots,x_n)$.  Thus, the protocol $P$ consists of the tuple $(R,S,A)$.  We say that the algorithm is \emph{$(\eps,\delta)$-differentially private in the shuffled model} if the algorithm
$
M_{R}(x_1,\dots,x_n) = S( \cup_{i=1}^{n} R(x_i) )
$
satisfies $(\eps,\delta)$-differential privacy. For more detail, see the discussion leading to Definition \ref{def:mixnet-dp}.

In contrast to the local model, differential privacy is now a property of all $n$ users' messages, and the $(\eps,\delta)$ may be functions of $n$.  However, if an adversary were to inject additional messages, then it would not degrade privacy, provided that those messages are independent of the honest users' data.
Thus, we may replace $n$, in our results, as a \emph{lower bound} on the number of honest users in the system.
For example, if we have a protocol that is private for $n$ users, but instead we have $\frac{n}{p}$ users of which we assume at least a $p$ fraction are honest, the protocol will continue to satisfy differential privacy.

\subsection{Algorithmic Results}

Our main result shows how to estimate any bounded, real-valued linear statistic (a \emph{statistical query}) in the shuffled model with error that nearly matches the best possible utility achievable in the central model.
\begin{thm} \label{thm:realsum}
For every $\eps \in (0,1)$, and every $\delta \gtrsim \eps n 2^{-\eps n}$ and every function $f \from \cX \to [0,1]$, there is a protocol $P$ in the shuffled model that is $(\eps,\delta)$-differentially private, and for every $n$ and every $X = (x_1,\dots,x_n) \in \cX^n$,
$$
\ex{}{\left| P(X) - \sum_{i=1}^{n} f(x_i) \right|} = O\left(\frac{1}{\eps} \log \frac{n}{\delta}\right).
$$
Each user sends $m = \Theta(\eps \sqrt{n})$ one-bit messages.
\end{thm}

For comparison, in the central model, the Laplace mechanism achieves $(\eps,0)$-differential privacy and error $O(\frac{1}{\eps})$.  In contrast, error $\Omega(\frac{1}{\eps}\sqrt{n})$ is necessary in the local model. Thus, for answering statistical queries, this protocol essentially has the best properties of the local and central models (up to logarithmic factors).

In the special case of estimating a sum of bits (or a Boolean-valued linear statistic), our protocol has a slightly nicer guarantee and form.
\begin{thm} \label{thm:bitsum}
For every $\eps \in (0,1)$, and every $\delta \gtrsim 2^{-\eps n}$ and every function $f \from \cX \to \zo$, there is a protocol $P$ in the shuffled model that is $(\eps,\delta)$-differentially private, and for every $n$ and every $X = (x_1,\dots,x_n) \in \cX^n$,
$$
\ex{}{\left| P(X) - \sum_{i=1}^{n} f(x_i) \right|} = O\left(\frac{1}{\eps} \sqrt{\log \frac{1}{\delta}} \right).
$$
Each user sends a single one-bit message.
\end{thm}

\noindent The protocol corresponding to Theorem~\ref{thm:bitsum} is extremely simple:
\begin{enumerate}
\item For some appropriate choice of $p \in (0,1)$, each user $i$ with input $x_i$ outputs $y_i = x_i$ with probability $1-p$ and a uniformly random bit $y_i$ with probability $p$.  When $\eps$ is not too small, $p \approx \frac{\log(1/\delta)}{\eps^2 n}$.
\item The analyzer collects the shuffled messages $y_1,\dots,y_n$ and outputs $$\frac{1}{1-p} \left(\sum_{i=1}^{n} y_i - \tfrac{p}{2}\right)\, .$$
\end{enumerate}

\myparagraph{Intuition.} In the local model, an adversary can map the set of observations $\{y_1,\dots, y_n\}$ to users.  Thus, to achieve $\eps$-differential privacy, the parameter $p$ should be set close to $\frac{1}{2}$. In our model, the attacker sees only the anonymized set of observations $\{y_1,\dots,y_n\}$, whose distribution can be simulated using only $\sum_i y_i$. Hence, to ensure that the protocol is differentially private, it suffices to ensure that $\sum_i y_i$ is private, which we show holds for $p \approx \frac{\log(1/\delta)}{\eps^2 n} \ll \frac{1}{2}$.\anote{New explanatory text} \jnote{Edited.}

\myparagraph{Communication Complexity.}
Our protocol for real-valued queries requires $\Theta(\eps \sqrt{n})$ bits per user.  In contrast, the local model requires just a single bit, but incurs error $\Omega(\frac{1}{\eps} \sqrt{n})$.  A generalization of Theorem~\ref{thm:realsum} gives error $O(\frac{\sqrt{n}}{r} + \frac{1}{\eps} \log \frac{r}{\delta})$ and sends $r$ bits per user, but we do not know if this tradeoff is necessary.  Closing this gap is an interesting open question.

\subsection{Negative Results}

We also prove negative results for algorithms in the \emph{one-message} shuffled model.  These results hinge on a structural characterization of private protocols in the one-message shuffled model.

\begin{thm}
\label{thm:remove-shuffler-informal}
If a protocol $P = (R,S,A)$ satisfies $(\eps,\delta)$-differential privacy in the one-message shuffled model, then $R$ satisfies $( \eps + \ln n, \delta)$-differential privacy.  Therefore, $P$ is $(\eps + \ln n, \delta)$-differentially private in the local model.
\end{thm} 

Using Theorem~\ref{thm:remove-shuffler-informal} (and a transformation of~\citep{bns18} from $(\eps,\delta)$-DP to $(O(\eps),0)$-DP in the local model), we can leverage existing lower bounds for algorithms in the local model to obtain lower bounds on algorithms in the shuffled model.

\myparagraph{Variable Selection.}
In particular, consider the following \emph{variable selection problem}: given a dataset $x \in \zo^{n \times d}$, output $\widehat J$ such that
$$
\sum_{i=1}^{n} x_{i,\widehat J} \geq \left( \max_{j \in [d]} \ \ \sum_{i=1}^{n} x_{i,j}\right) - \frac{n}{10} \, .
$$
(The $\frac n {10}$ approximation term is somewhat arbitrary---any sufficiently small constant fraction of $n$ will lead to the same lower bounds and separations.)

Any local algorithm (with $\eps = 1$) for selection requires $n = \Omega(d \log d)$, whereas in the central model the exponential mechanism~\cite{McSherryT07} solves this problem for $n = O(\log d)$. The following lower bound shows that for this ubiquitous problem, the one-message shuffled model cannot match the central model.

\begin{thm}
If $P$ is a $(1,\frac{1}{n^{10}})$-differentially private protocol in the one-message shuffled model that solves the selection problem (with high probability) then $n = \Omega(d^{1/17})$.  Moreover this lower bound holds even if $x$ is drawn iid from a product distribution over $\zo^{d}$.
\end{thm}

In Section~\ref{sec:lower-bounds}, we also prove lower bounds for the well studied \emph{histogram problem}, showing that any one-message shuffled-model protocol for this problem must have error growing (polylogarithmically) with the size of the data domain.  In contrast, in the central model it is possible to release histograms with no dependence on the domain size, even for \emph{infinite} domains.

We remark that our lower bound proofs do not apply if the algorithm sends multiple messages through the shuffler. However, we do not know whether beating the bounds is actually possible. Applying our bit-sum protocol $d$ times (together with differential privacy's composition property) shows that $n = \tilde{O}(\sqrt{d})$ samples suffice in the general shuffled model. We also do not know if this bound can be improved.  We leave it as an interesting direction for future work to fully characterize the power of the shuffled model.

\subsection{Comparison to~\cite{ErlingssonFMRTT19}}
\label{sec:comparison}
In concurrent and independent work, Erlingsson et al.~\cite{ErlingssonFMRTT19} give conceptually similar positive results for local protocols aided by a shuffler.  Specifically, they prove a general amplification result: adding a shuffler to any protocol satisfying local differential privacy improve the privacy parameters, often quite significantly.   This amplification result can be seen as a partial converse to our transformation from shuffled protocols to local protocols (Theorem~\ref{thm:remove-shuffler-informal}).

Their result applies to \emph{any} local protocol, whereas our protocol for bit-sums (Theorem~\ref{thm:bitsum}) applies specifically to the one-bit randomized response protocol.  However, when specialized to randomized response, their result is quantitatively weaker than ours.  As stated, their results only apply to local protocols that satisfy $\eps$-differential privacy for $\eps < 1$.  In contrast, the proof of Theorem~\ref{thm:bitsum} shows that, for randomized response, local differential privacy $\eps \approx \ln(n)$ can be amplified to $\eps' = 1$.  Our best attempt at generalizing their proof to the case of $\eps \gg 1$ does not give any amplification for local protocols with $\eps \approx \ln(n)$.  Specifically, our best attempt at applying their method to the case of randomized response yields a shuffled protocol that is  $1$-differentially private and has error $\Theta(n^{5/12})$, which is just slightly better than the error $O(\sqrt{n})$ that can be achieved without a shuffler.

\section{Model and Preliminaries}
In this section, we define terms and notation used throughout the paper.   We use $\mathrm{Ber}(p)$ to denote the Bernoulli distribution over $\zo$, which has value 1 with probability $p$ and 0 with probability $1-p$. We will use $\mathrm{Bin}(n,p)$ to denote the binomial distribution (i.e.~the sum of $n$ independent samples from $\mathrm{Ber}(p)$.

\subsection{Differential Privacy}
Let $X \in \cX^n$ be a \emph{dataset} consisting of elements from some universe $\cX$.  We say two datasets $X,X'$ are \emph{neighboring} if they differ on at most one user's data, and denote this $X \sim X'$.

\begin{defn}[Differential Privacy~\cite{DworkMNS06}] \label{def:dp}
An algorithm $M \from \cX^* \to \cZ$ is $(\eps,\delta)$-differentially private if for every $X \sim X' \in \cX^*$ and every $T \subseteq \cZ$
$$
\pr{}{M(X) \in T} \leq e^{\eps} \pr{}{M(X') \in T} + \delta.
$$
where the probability is taken over the randomness of $M$.
\end{defn}

Differential privacy satisfies two extremely useful properties:

\begin{lem}[Post-Processing~\cite{DworkMNS06}] \label{lem:postprocessing}
If $M$ is $(\eps,\delta)$-differentially private, then for every $A$, $A \circ M$ is $(\eps,\delta)$-differentially private.
\end{lem}

\begin{lem}[Composition~\cite{DworkMNS06,DworkRV10}] \label{lem:composition}
If $M_1,\dots,M_T$ are $(\eps,\delta)$-differentially private, then the composed algorithm
$$
\widetilde{M}(X) = (M_1(X),\dots,M_T(X))
$$
is $(\eps', \delta' + T\delta)$-differentially private for every $\delta' > 0$ and $\eps' = \eps(e^\eps -1)T + \eps \sqrt{2T \log(1/\delta')}$.
\end{lem}

\subsection{Differential Privacy in the Shuffled Model} \label{sec:model}

In our model, there are $n$ \emph{users}, each of whom holds data $x_i \in \cX$.  We will use $X = (x_1,\dots,x_n) \in \cX^n$ to denote the \emph{dataset} of all $n$ users' data.  We say two datasets $X,X'$ are \emph{neighboring} if they differ on at most one user's data, and denote this $X \sim X'$.

The protocols we consider consist of three algorithms:
\begin{itemize}
\item $R \from \cX \to \cY^m$ is a randomized \emph{encoder} that takes as input a single users' data $x_i$ and outputs a set of $m$ \emph{messages} $y_{i,1},\dots,y_{i,m} \in \cY$. If $m=1$, then $P$ is in the \emph{one-message shuffled model}.

\item $\shuff \from \cY^* \to \cY^*$ is a \emph{shuffler} that takes a set of messages and outputs these messages in a uniformly random order.  Specifically, on input $y_1,\dots,y_N$, $\shuff$ chooses a uniformly random permutation $\pi \from [N] \to [N]$ and outputs $y_{\pi(1)},\dots,y_{\pi(N)}$.

\item $A \from \cY^* \to \cZ$ is some \emph{analysis function} or \emph{analyzer} that takes a set of messages $y_1,\dots,y_N$ and attempts to estimate some function $f(x_1,\dots,x_n)$ from these messages.
\end{itemize}

\noindent We denote the overall protocol $P = (R,\shuff,A)$.  The mechanism by which we achieve privacy is $$\Pi_{R}(x_1,\dots,x_n) = \shuff(\cup_{i=1}^n R(x_i)) = \shuff(y_{1,1},\dots,y_{n,m}),$$ where both $R$ and $\shuff$ are randomized.  We will use $P(X) = A \circ \Pi_{R}(X)$ to denote the output of the protocol.  However, by the post-processing property of differential privacy (Lemma~\ref{lem:postprocessing}), it will suffice to consider the privacy of $\Pi_{R}(X)$, which will imply the privacy of $P(X)$.  We are now ready to define differential privacy for protocols in the shuffled model.  

\begin{defn}[Differential Privacy in the Shuffled Model]
\label{def:mixnet-dp}
A protocol $P = (R,\shuff,A)$ is $(\eps,\delta)$-differentially private if the algorithm $\Pi_{R}(x_1,\dots,x_n) = \shuff(R(x_1),\dots,R(x_n))$ is $(\eps,\delta)$-differentially private (Definition~\ref{def:dp}).
\end{defn}

In this model, privacy is a property of the entire set of users' messages and of the shuffler, and thus $\eps,\delta$ may depend on the number of users $n$.  When we wish to refer to $P$ or $\Pi$ with a specific number of users $n$, we will denote this by $P_{n}$ or $\Pi_{n}$.

We remark that if an adversary were to inject additional messages, then it would not degrade privacy, provided that those messages are independent of the honest users' data.
Thus, we may replace $n$, in our results, with an assumed \emph{lower bound} on the number of honest users in the system.

\medskip
In some of our results it will be useful to have a generic notion of accuracy for a protocol $P$.  
\begin{defn}[Accuracy of Distributed Protocols]
Protocol $P=(R,\shuff,A)$ is \emph{$(\alpha,\beta)$-accurate} for the function $f:\cX^* \rightarrow \cZ$ if, for every $X \in \cX^*$, we have $\pr{}{d(P(X),f(X)) \leq \alpha} \geq 1-\beta$ where $d \from \cZ \times \cZ \to \R$ is some application-dependent distance measure.
\end{defn}

As with the privacy guarantees, the accuracy of the protocol may depend on the number of users $n$, and we will use $P_n$ when we want to refer to the protocol with a specific number of users.

\myparagraph{Composition of Differential Privacy}
We will use the following useful composition property for protocols in the shuffled model, which is an immediate consequence of Lemma~\ref{lem:composition} and the post-processing Lemma~\ref{lem:postprocessing}.  This lemma allows us to directly compose protocols in the shuffled model while only using the shuffler once, rather than using the shuffler independently for each protocol being composed.
\begin{lem}[Composition of Protocols in the Shuffled Model] \label{ref:mixnetcomposition}
If $\Pi_1 = (R_1,\shuff),\dots,\Pi_{T} = (R_{T},\shuff)$ for $R_t \from \cX \to \cY^m$ are each $(\eps,\delta)$-differentially private in the shuffled model, and $\widetilde{R} \from \cX \to \cY^{mT}$ is defined as 
$$
\widetilde{R}(x_i) = (R_{1}(x_i),\dots,R_{T}(x_i))
$$
then, for every $\delta ' > 0$, the composed protocol $\widetilde{\Pi} = (\widetilde{R}, \shuff)$ is $(\eps', \delta' + T\delta)$-differentially private in the shuffled model for $\eps' = \eps(e^\eps - 1)T + \eps \sqrt{2T \log(1/\delta')}$.
\end{lem}

\subsubsection{Local Differential Privacy}

If the shuffler $\shuff$ were replaced with the identity function (i.e.~if it did not randomly permute the messages) then we would be left with exactly the \emph{local model of differential privacy}.  That is, a locally differentially private protocol is a pair of algorithms $P = (R,A)$, and the output of the protocol is $P(X) = A(R(x_1),\dots,R(x_n))$.  A protocol $P$ is differentially private in the local model if and only if the algorithm $R$ is differentially private.  In Section~\ref{sec:lower-bounds} we will see that if $P = (R,\shuff,A)$ is a differentially private protocol in the one-message shuffled model, then $R$ itself must satisfy local differential privacy for non-trivial $(\eps,\delta)$, and thus $(R,A \circ \shuff)$ is a differentially private local protocol for the same problem.


\section{A Protocol for Boolean Sums}
\label{sec:bit-sum}

In this section we describe and analyze a protocol for computing a sum of $\zo$ bits\ifnum\oakland=0, establishing Theorem~\ref{thm:bitsum} in the introduction\fi.

\subsection{The Protocol}

In our model, the data domain is $\cX = \zo$ and the function being computed is $f(x_1,\dots,x_n) = \sum_{i=1}^{n} x_i$.  Our protocol, $P_\lambda$, is specified by a parameter $\lambda \in [0,n]$ that allows us to trade off the level of privacy and accuracy.  Note that $\lambda$ may be a function of the number of users $n$.  We will discuss in Section~\ref{sec:plambda-setting} how to set this parameter to achieve a desired level of privacy.  For intuition, one may wish to think of the parameter $\lambda \approx \frac{1}{\eps^2}$ when $\eps$ is not too small.

The basic outline of $P_\lambda$ is as follows.  Roughly, a random set of $\lambda$ users will choose $y_i$ randomly, and the remaining $n - \lambda$ will choose $y_i$ to be their input bit $x_i$.  The output of each user is the single message $y_i$.  The outputs are then shuffled and the output of the protocol is the sum $\sum_{i=1}^{n} y_i$, shifted and scaled so that it is an unbiased estimator of $\sum_{i=1}^{n} x_i$. 

The protocol is described in Algorithm~\ref{alg:bit-sum-shuffled}. The full name of this protocol is $P^{0/1}_\lambda$, where the superscript serves to distinguish it with the real sum protocol $P^\R_{\lambda,r}$ (Section \ref{sec:real-sum}). Because of the clear context of this section, we drop the superscript. Since the analysis of both the accuracy and utility of the algorithm will depend on the number of users $n$, we will use $P_{n,\lambda}, R_{n,\lambda}, A_{n,\lambda}$ to denote the protocol and its components in the case where the number of users is $n$.


\begin{algorithm}
\ifnum\oakland=1
\caption{The protocol $P^{0/1}_{n,\lambda}=(R^{0/1}_{n,\lambda}, \shuff, A^{0/1}_{n,\lambda})$}
\else
\caption{A shuffled protocol $P^{0/1}_{n,\lambda}=(R^{0/1}_{n,\lambda}, \shuff, A^{0/1}_{n,\lambda})$ 
for computing the sum of bits}
\fi
\label{alg:bit-sum-shuffled}

\tcp{Local Randomizer}
\Fn{$R^{0/1}_{n,\lambda}(x)$}{
	\KwIn{$x\in \zo$, parameters $n \in \N, \lambda\in(0,n)$.}
	\KwOut{$\mathbf{y} \in \zo$} \vspace{5pt}
    
	Let $\mathbf{b} \gets \Ber(\frac{\lambda}{n})$ \\
	\lIf{$\mathbf{b} = 0$}{ \Return{$\mathbf{y} \gets x$} }
	\lElseIf{$\mathbf{b} = 1$}{ \Return{$\mathbf{y} \gets \Ber\left(\half\right)$} }
}

\vspace{10pt}

\tcp{Analyzer}
\Fn{$A^{0/1}_{n,\lambda}(y_1,\dots,y_n)$}{
	\KwIn{$(y_1,\dots,y_n)\in \zo^n$, parameters $n \in \N, \lambda\in(0,n)$.}
	\KwOut{$z\in [0,n]$} \vspace{5pt}
    
	\Return{$z\gets \frac{n}{n-\lambda} \cdot \left( \sum_{i=1}^{n} y_i - \frac{\lambda}{2} \right)$}
}
\end{algorithm}


\subsection{Privacy Analysis} \label{sec:plambda-privacy}

In this section we will prove that $P_{\lambda}$ satisfies $(\eps,\delta)$-differential privacy.  Note that if $\lambda = n$ then the each user's output is independent of their input, so the protocol trivially satisfies $(0,0)$-differential privacy, and thus our goal is to prove an upper bound on the parameter $\lambda$ that suffices to achieve a given $(\eps,\delta)$.

\begin{thm}[Privacy of $P_{\lambda}$]\label{thm:bitsumprivacy}
There are absolute constants $\kappa_1,\dots,\kappa_5$ such that the following holds for $P_\lambda$. For every $n \in \N$, $\delta \in (0,1)$ and $\frac{\kappa_2 \log(1/\delta)}{n}\leq \eps \leq 1$, there exists a $\lambda = \lambda(n,\eps,\delta)$ such that $P_{n,\lambda}$ is $(\eps,\delta)$ differentially private and, 
\begin{equation*}
\lambda \leq 
\begin{cases}
\frac{\kappa_4 \log(1/\delta)}{\eps^2} & \textrm{if $\eps \geq \sqrt{\frac{\kappa_3 \log(1/\delta)}{n}}$} \\
n - \frac{\kappa_5 \eps n^{3/2}}{\sqrt{\log(1/\delta)}} & \textrm{otherwise}
\end{cases}
\end{equation*}
\end{thm}

\noindent In the remainder of this section we will prove Theorem~\ref{thm:bitsumprivacy}.

The first step in the proof is the observation that the output of the shuffler depends only on $\sum_i y_i$. It will be more convenient to analyze the algorithm $C_{\lambda}$ (Algorithm \ref{alg:bit-sum}) that simulates $\shuff(R_\lambda(x_1),\dots,R_{\lambda}(x_n))$. Claim~\ref{clm:equivalence-to-C_lambda} shows that the output distribution of $C_{\lambda}$ is indeed the same as that of the output $\sum_{i} y_i$.  Therefore, privacy of $C_{\lambda}$ carries over to $P_{\lambda}$.


\begin{algorithm}
\caption{$C_{\lambda}(x_1\dots x_n)$}
\label{alg:bit-sum}
	\KwIn{$(x_1\dots x_n) \in \zo^n$, parameter $\lambda\in(0,n)$.}
	\KwOut{$\mathbf{y} \in \{0,1,2,\dots,n\}$} \vspace{5pt}
    
  Sample $\mathbf{s} \gets \Bin\left(n,\frac{\lambda}{n}\right)$ \\
  Define $\mathcal{H}_s = \{ H \subseteq [n] : |H| = s\}$ and choose $\mathbf{H} \gets \mathcal{H}_{s}$ uniformly at random \\
  \Return{$\mathbf{y} \gets \sum_{i\notin \mathbf{H}} x_i + \Bin\left(s,\half\right)$ }
\end{algorithm}

\begin{clm}
\label{clm:equivalence-to-C_lambda}
For every $n\in \N$, $x\in \zo^n$, and every $r\in\{0,1,2,\dots,n\}$, \[
\pr{}{C_{\lambda}(X)=r}=\pr{}{\sum_{i=1}^n R_{n,\lambda}(x_i) =r}\]
\end{clm}

\ifnum\oakland=0
\begin{proof}
Fix any $r\in \{0,1,2,\dots,n\}$. \anote{How should we lead into this calculation?}
\begin{align*}
\pr{}{C_{\lambda}(X)=r} &= \sum_{H\subseteq[n]} \pr{}{C_{\lambda}(X)=r \cap \mathbf{H}=H}\\
&= \sum_{H\subseteq [n]} \pr{}{\sum_{i\notin H} x_i + \mathit{Bin}\left(|H|,\half\right) =r} \cdot\left(\frac{\lambda}{n}\right)^{|H|}\left(1-\frac{\lambda}{n}\right)^{n-|H|}\\
&= \sum_{H\subseteq [n]} \pr{}{\sum_{i\notin H} x_i + \sum_{i\in H} \mathit{Ber}\left(\half\right) =r} \cdot\left(\frac{\lambda}{n}\right)^{|H|}\left(1-\frac{\lambda}{n}\right)^{n-|H|} \stepcounter{equation} \label{eq:breakdown-Clambda} \tag{\theequation}
\end{align*}
\noindent Let $\mathbf{G}$ denote the (random) set of people for whom $b_i= 1$ in $P_{\lambda}$. Notice that
\begin{align*}
\pr{}{\sum_{i=1}^n R_{n,\lambda}(x_i)=r} &= \sum_{G \subseteq [n]} \pr{}{\sum_i R_{n,\lambda}(x_i) = r \cap \mathbf{G} = G}\\
&= \sum_{G\subseteq [n]} \pr{}{\sum_{i\notin G} x_i + \sum_{i\in G}\mathit{Ber}\left(\half\right) =r} \\
&~~~~~\cdot\left(\frac{\lambda}{n}\right)^{|G|}\left(1-\frac{\lambda}{n}\right)^{n-|G|}
\end{align*}
which is the same as \eqref{eq:breakdown-Clambda}. This concludes the proof.
\end{proof}
\fi

Now we establish that in order to demonstrate privacy of $P_{n,\lambda}$, it suffices to analyze $C_{\lambda}$.
\begin{clm}
\label{clm:sufficient-for-privacy}
If $C_{\lambda}$ is $(\eps,\delta)$ differentially private, then $P_{n,\lambda}$ is $(\eps,\delta)$ differentially private.
\end{clm}

\ifnum\oakland=0
\begin{proof}
Fix any number of users $n$.  Consider the randomized algorithm $T:\{0,1,2,\dots,n\} \rightarrow \bits^n$ that takes a number $r$ and outputs a uniformly random string $z$ that has $r$ ones. If $C_{\lambda}$ is differentially private, then the output of $T\circ C_{\lambda}$ is $(\eps,\delta)$ differentially private by the post-processing lemma.

To complete the proof, we show that for any $X\in\cX^n$ the output of $(T\circ C_{\lambda})(X)$ has the same distribution as $\shuff(R_{\lambda}(x_1),\dots R_{\lambda}(x_n))$.  Fix some vector $Z \in \zo^{n}$ with sum $r$
\begin{align*}
\pr{T,C_{\lambda}}{T(C_{\lambda}(X))=Z} &= \pr{}{T(r)=Z}\cdot \pr{}{C_{\lambda}(X)=r}\\
&= \textstyle\binom{n}{r}^{-1}\cdot \pr{}{C_{\lambda}(X)=r} \\
&= \textstyle\binom{n}{r}^{-1}\cdot \pr{}{f(R_{n,\lambda}(X))=r} \tag{Claim \ref{clm:equivalence-to-C_lambda}}\\
&= {\textstyle\binom{n}{r}^{-1}}\cdot \sum_{Y\in \zo^{n} : |Y| = r} \pr{}{R_{n,\lambda}(X)=Y} \\
&= \sum_{Y\in \zo^{n} : |Y| = r} \pr{}{R_{n,\lambda}(X)=Y}\cdot \pr{}{\shuff(Y)=Z}\\
&= \pr{R_{n,\lambda},\shuff}{\shuff(R_{n,\lambda}(X))=Z}
\end{align*}
This completes the proof of Claim \ref{clm:sufficient-for-privacy}.
\end{proof}
\fi

We will analyze the privacy of $C_{\lambda}$ in three steps.  First we show that for \emph{any} sufficiently large $H$, the final step (encapsulated by Algorithm \ref{bit-sum-central-H}) will ensure differential privacy for some parameters.  When then show that for \emph{any} sufficiently large value $s$ and $H$ chosen \emph{randomly} with $|H| = s$, the privacy parameters actually improve significantly in the regime where $s$ is close to $n$; this sampling of $H$ is performed by Algorithm \ref{bit-sum-central-s}.  Finally, we show that when $s$ is chosen \emph{randomly} then $s$ is sufficiently large with high probability.

\begin{algorithm}
\caption{$C_H(x_1\dots x_n)$}
\label{bit-sum-central-H}

	\KwIn{$(x_1\dots x_n)\in \zo^n$, parameter $H\subseteq [n]$.}
	\KwOut{$\mathbf{y}_H \in \{0,1,2,\dots,n\}$} \vspace{5pt}

Let $\mathbf{B}\leftarrow \mathit{Bin}\left(|H|,\half\right)$ \\
\Return{$\mathbf{y}_H \gets \sum_{i\notin H} x_i + \mathbf{B}$}
\end{algorithm}

\begin{clm}
\label{clm:C_H}
For any $\delta>0$ and any $H\subseteq [n]$ such that $|H|>8\log \frac{4}{\delta}$, $C_H$ is $(\eps,\frac{\delta}{2})$-differentially private for
\ifnum\oakland=0
\[
\eps = \ln\left(1+\sqrt{\frac{32 \log\frac{4}{\delta}}{|H|}}\right) < \sqrt{\frac{32 \log\frac{4}{\delta}}{|H|}}
\]
\else
$\eps \leq \sqrt{\frac{32 \log\frac{4}{\delta}}{|H|}}$
\fi
\end{clm}
\begin{proof}
Fix neighboring datasets $X\sim X'\in \zo^n$, any $H\subseteq [n]$ such that $|H|> 8\log\frac{4}{\delta}$, and any $\delta>0$. If the point at which $X,X'$ differ lies within $H$, the two distributions $C_H(X),C_H(X')$ are identical. Hence, without loss of generality we assume that $x_j=0$ and $x_j'=1$ for some $j \not\in H$.


Define $u:=\sqrt{\half|H|\log\frac{4}{\delta}}$ and $I_u:= \left(\half|H|-u, \half|H|+u \right)$ so that by Hoeffding's inequality\ifnum\makeappendix=1 (Theorem \ref{thm:hoeffding})\fi, $\pr{}{\mathbf{B} \in I_u} < \half\delta$. For any $W\subseteq \{0,1,2,\dots,n\}$ we have,
\ifnum\oakland=1
\begin{align*}
\pr{}{C_H(X)\in W}
&\leq \pr{}{C_H(X)\in W \cap \mathbf{B} \in I_u} + \tfrac12 \delta\\
&= \sum_{r \in W\cap I_u} \mathbb{P}\big[\mathbf{B}+ \textstyle\sum_{i\notin H} x_i = r \big] + \half \delta
\end{align*}
\else
\begin{align*}
\pr{}{C_H(X)\in W} &= \pr{}{C_H(X)\in W \cap \mathbf{B} \in I_u} + \pr{}{C_H(X)\in W \cap \mathbf{B} \notin I_u}\\
&\leq \pr{}{C_H(X)\in W \cap \mathbf{B} \in I_u} + \half \delta\\
&= \sum_{r \in W\cap I_u} \pr{}{\mathbf{B}+\sum_{i\notin H} x_i = r} + \half \delta  
\end{align*}
\fi

\noindent Thus to complete the proof, it suffices to show that for any $H$ and $r \in W\cap I_u$
\begin{equation}
\frac{\pr{}{\mathbf{B}+\sum_{i\notin H} x_i =r}}{\pr{}{\mathbf{B}+\sum_{i\notin H} x_i' =r}} \leq 1+\sqrt{\frac{32 \log\frac{4}{\delta}}{|H|}} \label{eq:C_H-core}
\end{equation}

Because $x_j=0,x_j'=1$ and $j\notin H$, we have $\sum_{i\notin H} x_i = \sum_{i\notin H} x_i' - 1$. Thus,
\ifnum\oakland=1
\begin{align*}
\frac{\pr{}{\mathbf{B}+\sum_{i\notin H} x_i =r}}{\pr{}{\mathbf{B}+\sum_{i\notin H} x_i' =r}} = \frac{\pr{}{\mathbf{B}= \left(r-\sum_{i\notin H} x_i'\right)+1}}{\pr{}{\mathbf{B} = \left(r-\sum_{i\notin H} x_i'\right)}}
\end{align*}
\else
\begin{multline*}
\frac{\pr{}{\mathbf{B}+\sum_{i\notin H} x_i =r}}{\pr{}{\mathbf{B}+\sum_{i\notin H} x_i' =r}} = \frac{\pr{}{\mathbf{B}+\sum_{i\notin H} x_i' - 1 =r}}{\pr{}{\mathbf{B}+\sum_{i\notin H} x_i' =r}} 
\\
= \frac{\pr{}{\mathbf{B}= \left(r-\sum_{i\notin H} x_i'\right)+1}}{\pr{}{\mathbf{B} = \left(r-\sum_{i\notin H} x_i'\right)}}
\end{multline*}
\fi
Now we define $k = r-\sum_{i\notin H} x_i'$ so that $$\frac{\pr{}{\mathbf{B}= \left(r-\sum_{i\notin H} x_i'\right)+1}}{\pr{}{\mathbf{B} = \left(r-\sum_{i\notin H} x_i'\right)}} = \frac{\pr{}{\mathbf{B} = k+1}}{\pr{}{\mathbf{B} = k}}.$$
\ifnum\oakland=1
The remainder of the proof is a calculation involving the binomial distribution and the parameters $u,|H|$,
\begin{align*}
\frac{\pr{}{\mathbf{B} = k+1}}{\pr{}{\mathbf{B} = k}} 
= \frac{|H|-k}{k+1}
&\leq \frac{|H|-(\half|H|-u)}{\half|H|-u+1} \\
&\leq 1+ \sqrt{\frac{32\log\frac{4}{\delta}}{|H|}}
\end{align*}
\else
Then we can calculate
\begin{align*}
\frac{\pr{}{\mathbf{B} = k+1}}{\pr{}{\mathbf{B} = k}} 
&= \frac{|H|-k}{k+1} \tag{$\mathbf{B}$ is binomial}\\
&\leq \frac{|H|-(\half|H|-u)}{\half|H|-u+1} \tag{$r \in I_{u}$ so $k \geq \half |H| - u$}\\
&< \frac{\half|H|+u}{\half|H|-u}=\frac{u^2/(\log\frac{4}{\delta})+u}{u^2/(\log\frac{4}{\delta})-u} \tag{$u = \sqrt{\half|H|\log\frac{4}{\delta}}$}   \\
&=\frac{u+\log\frac{4}{\delta}}{u-\log\frac{4}{\delta}} =1 + \frac{2\log \frac{4}{\delta}}{u-\log\frac{4}{\delta}} = 1 + \frac{2\log \frac{4}{\delta}}{\sqrt{\half|H|\log \frac{4}{\delta}}-\log\frac{4}{\delta}}\\
&\leq 1+ \frac{4\log \frac{4}{\delta}}{\sqrt{\half|H|\log \frac{4}{\delta}}} = 1+ \sqrt{\frac{32\log\frac{4}{\delta}}{|H|}} \tag{$|H| > 8\log\frac{4}{\delta}$}
\end{align*}
\fi
which completes the proof.
\end{proof}

Next, we consider the case where $H$ is a \emph{random} subset of $[n]$ with a \emph{fixed} size $s$.  In this case we will use an \emph{amplification via sampling argument}~\cite{KasiviswanathanLNRS08,adamSampleSecrecy} to argue that the randomness of $H$ improves the privacy parameters by a factor of roughly $(1-\frac{s}{n})$, which will be crucial when $s \approx n$. 
\ifnum\oakland=1 We omit the proof for space. \fi

\begin{algorithm}
\caption{$C_s(x_1,\dots,x_n)$}
\label{bit-sum-central-s}

	\KwIn{$(x_1,\dots,x_n)\in \zo^n$, parameter $s\in \{0,1,2,\dots,n\}$.}
	\KwOut{$\mathbf{y}_s \in \{0,1,2,\dots,n\}$} \vspace{5pt}

  Define $\mathcal{H}_s = \{ H \subseteq [n] : |H| = s\}$ and choose $\mathbf{H} \gets \mathcal{H}_{s}$ uniformly at random\\
  \Return{$\mathbf{y}_s \gets C_\mathbf{H}(x)$}
\end{algorithm}

\begin{clm}
\label{clm:C_s}
For any $\delta>0$ and any $s>8\log\frac{4}{\delta}$, $C_s$ is $( \eps ,\half\delta)$ differentially private for
$$
\eps = \sqrt{\frac{32\log\frac{4}{\delta}}{s}} \cdot \left(1-\frac{s}{n}\right)
$$
\end{clm}

\ifnum\oakland=0

\begin{proof}
As in the previous section, fix $X\sim X'\in \zo^n$ where $x_j=0, x_j'=1$. $C_s(X)$ selects $\mathbf{H}$ uniformly from $\mathcal{H}_s$ and runs $C_H(X)$; let $H$ denote the realization of $\mathbf{H}$. To enhance readability, we will use the shorthand $\eps_0(s) := \sqrt{\frac{32\log\frac{4}{\delta}}{s}}$. For any $W\subset \{0, 1, 2, \dots, n\}$, we aim to show that
\[
\frac{\pr{\mathbf{H},C_\mathbf{H}}{C_\mathbf{H}(X)\in W}-\half\delta}{\pr{\mathbf{H},C_\mathbf{H}}{C_\mathbf{H}(X')\in W}} \leq \exp\left(\eps_0(s) \cdot \left(1-\frac{s}{n}\right) \right)
\]

First, we have 
\begin{align}
&\frac{\pr{\mathbf{H},C_\mathbf{H}}{C_\mathbf{H}(X)\in W}-\half\delta}{\pr{\mathbf{H},C_\mathbf{H}}{C_\mathbf{H}(X')\in W}} \notag \\
={} &\frac{ \pr{\mathbf{H},C_\mathbf{H}}{C_\mathbf{H}(X)\in W \cap j\in \mathbf{H}} + \pr{\mathbf{H},C_\mathbf{H}}{C_\mathbf{H}(X)\in W\cap j\notin \mathbf{H}} -\half\delta}{\pr{\mathbf{H},C_\mathbf{H}}{C_\mathbf{H}(X')\in W\cap j\in \mathbf{H}} + \pr{\mathbf{H},C_\mathbf{H}}{C_\mathbf{H}(X')\in W\cap j\notin \mathbf{H}} }  \notag \\
={} &\frac{(1-p)\gamma(X)+p\zeta(X)-\half\delta}{(1-p)\gamma(X')+p\zeta(X')} \label{eq:amplification1}
\end{align}

where $p := \pr{}{j\notin \mathbf{H}}=(1-s/n)$, $$\gamma(X) := \pr{C_\mathbf{H}}{C_\mathbf{H}(X)\in W\mid j\in \mathbf{H}}\quad\textrm{and}\quad\zeta(X) := \pr{C_\mathbf{H}}{C_\mathbf{H}(X)\in W\mid j\notin \mathbf{H}}\, .$$
When user $j$ outputs a uniformly random bit, their private value has no impact on the distribution.  Hence, $\gamma(X)=\gamma(X')$, and 
\begin{equation} \label{eq:amplification2}
\eqref{eq:amplification1} = \frac{(1-p)\gamma(X)+p\zeta(X)-\half\delta}{(1-p)\gamma(X)+p\zeta(X')}
\end{equation}
Since $s = |H|$ is sufficiently large, by Claim~\ref{clm:C_H} we have $\zeta(X) \leq (1+\eps_0(s)) \cdot \min\{ \zeta(X'), \gamma(X) \} + \half\delta$.
\begin{align}
\eqref{eq:amplification2} &\leq \frac{(1-p)\gamma(X)+p\cdot(1+\eps_0(s)) \cdot \min\{\zeta(X'),\gamma(X)\}+\delta)-\half\delta}{(1-p)\gamma(X)+p\zeta(X')} \notag \\
&\leq \frac{(1-p)\gamma(X)+p\cdot (1+\eps_0(s)) \cdot\min\{\zeta(X'),\gamma(X)\}}{(1-p)\gamma(X)+p\zeta(X')} \notag \\
&= \frac{(1-p)\gamma(X)+p\cdot \min (\zeta(X'),\gamma(X)) + p\cdot \eps_0(s) \cdot \min\{\zeta(X'),\gamma(X)\}}{(1-p)\gamma(X)+p\zeta(X')} \notag \\
&\leq \frac{(1-p)\gamma(X)+p\zeta(X') + p\cdot \eps_0(s) \cdot \min\{\zeta(X'),\gamma(X)\}}{(1-p)\gamma(X)+p\zeta(X')} \notag \\
&= 1+ \frac{p\cdot \eps_0(s) \cdot \min\{\zeta(X'),\gamma(X)\}}{(1-p)\gamma(X)+p\zeta(X')} \label{eq:amplification3}
\end{align}

Observe that $\min\{\zeta(X'),\gamma(X)\} \leq (1-p) \gamma(X) + p \zeta(X')$, so
\begin{align*}
\eqref{eq:amplification3} \leq 1+ p\cdot \eps_0(s)
= 1+\eps_0(s)\cdot\left(1-\frac{s}{n}\right)
&\leq \exp\left( \eps_0(s)\cdot \left(1-\frac{s}{n}\right) \right)\\
&= \exp\left( \sqrt{\frac{32\log(4/\delta)}{s}} \cdot \left(1-\frac{s}{n}\right) \right)
\end{align*}
which completes the proof.
\end{proof}
\else
\fi

We now come to the actual algorithm $C_{\lambda}$, where $s$ is not fixed but is random.  The analysis of $C_{s}$ yields a bound on the privacy parameter that increases with $s$, so we will complete the analysis of $C_{\lambda}$ by using the fact that, with high probability, $s$ is almost as large as $\lambda$.  

\begin{clm}
\label{clm:C_lambda}
For any $\delta>0$ and $n \geq \lambda \geq 14\log\frac{4}{\delta}$, $C_{\lambda}$ is $(\eps, \delta )$ differentially private where
\[
\eps = \sqrt{\frac{32\log\frac{4}{\delta}}{\lambda- \sqrt{2\lambda\log\tfrac{2}{\delta}}}} \cdot \left( 1 - \frac{\lambda- \sqrt{2 \lambda\log\tfrac{2}{\delta}}}{n} \right)
\]
\end{clm}

The proof is in
\ifnum\makeappendix=1
\ifnum\eurocrypt=0 Appendix \else Supplementary Material \fi \ref{sec:appendix-bitsum-privacy}.
\else the full version of the paper\fi

\medskip

From Claim \ref{clm:sufficient-for-privacy}, $C_{\lambda}$ and $P_{n,\lambda}$ share the same privacy guarantees. Hence, Claim \ref{clm:C_lambda} implies the following:
\begin{coro}
\label{coro:P_lambda-private}
For any $\delta\in(0,1)$, $n \in \N$, and $\lambda \in \left[14\log\frac{4}{\delta},n\right]$, $P_{n,\lambda}$ is $(\eps,\delta)$ differentially private, where
\[
\eps = \sqrt{\frac{32\log\frac{4}{\delta}}{\lambda- \sqrt{2\lambda\log\tfrac{2}{\delta}}}} \cdot \left( 1 - \frac{\lambda- \sqrt{2 \lambda\log\tfrac{2}{\delta}}}{n} \right)
\]
\end{coro}

\subsection{Setting the Randomization Parameter} \label{sec:plambda-setting}

Corollary~\ref{coro:P_lambda-private} gives a bound on the privacy of $P_{n,\lambda}$ in terms of the number of users $n$ and the randomization parameter $\lambda$.  While this may be enough on its own, in order to understand the tradeoff between $\eps$ and the accuracy of the protocol, we want to identify a suitable choice of $\lambda$ to achieve a desired privacy guarantee $(\eps,\delta)$.  To complete the proof of Theorem~\ref{thm:bitsumprivacy}, we prove such a bound.

For the remainder of this section, fix some $\delta \in (0,1)$.  Corollary~\ref{coro:P_lambda-private} states that for any $n$ and $\lambda \in \left[14\log\frac{4}{\delta},n\right]$, $P_{n,\lambda}$ satisfies $(\eps^*(\lambda),\delta)$-differential privacy, where
\begin{align*}
\eps^*(\lambda) = \sqrt{\frac{32\log\frac{4}{\delta}}{\lambda- \sqrt{2\lambda\log\tfrac{2}{\delta}}}} \cdot \left( 1 - \frac{\lambda- \sqrt{2 \lambda\log\tfrac{2}{\delta}}}{n} \right)
\end{align*}
Let $\lambda^*(\eps)$ be the inverse of $\eps^*$, i.e.~the minimum $\lambda \in [0,n]$ such that $\eps^*(\lambda) \leq \eps$.  Note that $\eps^*(\lambda)$ is decreasing as $\lambda \to n$ while $\lambda^*(\eps)$ increases as $\eps \to 0$. By definition, $P_{n,\lambda}$ satisfies $(\eps,\delta)$ privacy if $\lambda \geq \lambda^*(\eps)$; the following Lemma gives such an upper bound:


\begin{lem}
\label{lem:generate-lambda}
For all $\delta \in (0,1)$, $n \geq 14\log\frac{4}{\delta}$, $\eps \in \left(\frac{\sqrt{3456}}{n}\log\frac{4}{\delta},1\right)$, $P_{n,\lambda}$ is $(\eps,\delta)$ differentially private if
\begin{equation}
\lambda = \begin{cases}
\frac{64}{\eps^2}\log\frac{4}{\delta} & \textrm{if $\eps \geq \sqrt{\frac{192}{n}\log\frac{4}{\delta}}$}\\
n-\frac{\eps n^{3/2}}{\sqrt{432\log (4/\delta)}} & \textrm{otherwise}
\end{cases}
\end{equation}
\end{lem}

\ifnum\oakland=1
The proof is essentially a case analysis depending on which of the two terms in $\eps^*(\lambda)$ is dominant.  When $\lambda < n/2$ we have that $\eps^*(\lambda) \approx \sqrt{1/\lambda}$, which gives the choice of $\lambda$ for relatively large values of $\eps$, and when $\lambda \geq n/2$ we have that $\eps^*(\lambda) \approx \frac{1}{\sqrt{n}}(1-\frac{\lambda}{n})$, which gives the choice of $\lambda$ for relatively small values of $\eps$.
\else
We'll prove the lemma in two claims, each of which corresponds to one of the two cases of our bound on $\lambda^*(\eps)$.  The first bound applies when $\eps$ is relatively large.


\begin{clm}
\label{clm:lambda-bound-1}
For all $\delta \in (0,1)$, $n \geq 14\log\frac{4}{\delta}$, $\eps \in \left(\sqrt{\frac{192}{n}\log\frac{4}{\delta}},1\right)$, if $\lambda = \frac{64}{\eps^2} \log \frac{4}{\delta}$ then $P_{n,\lambda}$ is $(\eps,\delta)$ private.
\end{clm}

\begin{proof}
Let $\lambda = \frac{64}{\eps^2} \log\frac{4}{\delta}$ as in the statement. Corollary~\ref{coro:P_lambda-private} states that $P_{n,\lambda}$ satisfies $(\eps^*(\lambda),\delta)$ privacy for
\begin{align*}
\eps^*(\lambda) &= \sqrt{\frac{32\log\frac{4}{\delta}}{\lambda- \sqrt{2\lambda\log\tfrac{2}{\delta}}}} \cdot \left( 1 - \frac{\lambda- \sqrt{2 \lambda\log\tfrac{2}{\delta}}}{n} \right) \\
&\leq \sqrt{\frac{32\log\frac{4}{\delta}}{\lambda- \sqrt{2\lambda\log\tfrac{2}{\delta}}}} \tag{$\lambda \leq n$} \\
&\leq \sqrt{\frac{64\log\frac{4}{\delta}}{\lambda}} \tag{$\lambda \geq 8 \log \frac{2}{\delta}$} \\
&= \eps
\end{align*}
This completes the proof of the claim.
\end{proof}

The value of $\lambda$ in the previous claim can be as large as $n$ when $\eps$ approaches $1/\sqrt{n}$.  We now give a meaningful bound for smaller values of $\eps$.

\begin{clm}
\label{clm:lambda-bound-2}
For all $\delta \in (0,1)$, $n \geq 14\log\frac{4}{\delta}$, $\eps \in \left(\frac{\sqrt{3456}}{n}\log\frac{4}{\delta}, \sqrt{\frac{192}{n} \log \frac{4}{\delta}}\right)$, if $$\lambda= n - \frac{\eps n^{3/2}}{\sqrt{432\log(4/\delta)}}$$ then $P_{n,\lambda}$ is $(\eps,\delta)$ private.
\end{clm}

\begin{proof}
Let $\lambda = n - \eps n^{3/2} / \sqrt{432\log(4/\delta)}$ as in the statement. Note that for this $\eps$ regime, we have $n/3 < \lambda < n$. Corollary~\ref{coro:P_lambda-private} states that $P_{n,\lambda}$ satisfies $(\eps^*(\lambda),\delta)$ privacy for

\begin{align*}
\eps^*(\lambda) &= \sqrt{\frac{32\log\frac{4}{\delta}}{\lambda- \sqrt{2\lambda\log\tfrac{2}{\delta}}}} \cdot \left( 1 - \frac{\lambda- \sqrt{2 \lambda\log\tfrac{2}{\delta}}}{n} \right)\\
&\leq \sqrt{\frac{64\log\frac{4}{\delta}}{\lambda}} \cdot \left( 1 - \frac{\lambda- \sqrt{2 \lambda\log\tfrac{2}{\delta}}}{n} \right) \tag{$\lambda \geq 8\log\frac{2}{\delta}$} \\
&= \sqrt{\frac{64\log\frac{4}{\delta}}{\lambda}} \cdot \left( \frac{\eps \sqrt{n}}{\sqrt{432\log(4/\delta)}} + \frac{\sqrt{2 \lambda\log\tfrac{2}{\delta}}}{n} \right) \\
&\leq \sqrt{\frac{64\log\frac{4}{\delta}}{\lambda}} \cdot \left( \frac{\eps \sqrt{n}}{\sqrt{432\log(4/\delta)}} + \sqrt{ \frac{2 \log\tfrac{2}{\delta}}{n}} \right) \tag{$\lambda \leq n$}\\
&\leq \sqrt{\frac{192\log\frac{4}{\delta}}{n}} \cdot \left( \frac{\eps \sqrt{n}}{\sqrt{432\log(4/\delta)}} + \sqrt{ \frac{2 \log\tfrac{2}{\delta}}{n}} \right) \tag{$\lambda \geq n/3$}\\
&= \frac{2}{3}\eps + \frac{\sqrt{384 \log\frac{4}{\delta} \log\tfrac{2}{\delta}}}{n} < \frac{2}{3}\eps + \frac{\sqrt{384}}{n}\log\frac{4}{\delta}\\
&< \frac{2}{3}\eps + \frac{1}{3}\eps = \eps \tag{$\eps > \frac{\sqrt{3456}}{n}\log\frac{4}{\delta}$}
\end{align*}
which completes the proof.
\end{proof}
\fi

\subsection{Accuracy Analysis} \label{sec:plambda-accuracy}

In this section, we will bound the error of $P_{\lambda}(X)$ with respect to $\sum_i x_i$.  Recall that, to clean up notational clutter, we will often write $f(X) = \sum_i x_i$. As with the previous section, our statements will at first be in terms of $\lambda$ but the section will end with a statement in terms of $\eps,\delta$.

\begin{thm}
\label{thm:bit-sum-accurate}
For every $n \in \N$, $\beta > 0$, $n > \lambda \geq 2\log\frac{2}{\beta}$, and $x \in \zo^{n}$,
$$
\pr{}{\left| P_{n,\lambda}(x) - \sum_i x_i \right| > \sqrt{2\lambda \log(2/\beta)} \cdot \left(\frac{n}{n-\lambda}\right)} \leq \beta
$$
\end{thm}

Observe that, using the choice of $\lambda$ specified in Theorem~\ref{thm:bitsumprivacy}, we conclude that for every $\frac1n \lesssim \eps \lesssim 1$ and every $\delta$ the protocol $P_{\lambda}$ satisfies
$$
\pr{}{\left| P_{n,\lambda}(x) - \sum_i x_i \right| > O\left( \frac{\sqrt{\log(1/\delta)\log(1/\beta)}}{\eps} \right)} \leq \beta
$$
To see how this follows from Theorem~\ref{thm:bit-sum-accurate}, consider two parameter regimes:
\begin{enumerate}
\item When $\eps \gg 1/\sqrt{n}$ then $\lambda \approx \frac{\sqrt{\log(1/\delta)}}{\eps^2} \ll n$, so the bound in Theorem~\ref{thm:bit-sum-accurate} is $O(\sqrt{\lambda \log(1/\beta)})$, which yields the desired bound.  
\item When $\eps \ll 1/\sqrt{n}$ then $n-\lambda \approx \eps n^{3/2}/\sqrt{\log(1/\delta)} \ll n$, so the bound in Theorem~\ref{thm:bit-sum-accurate} is $O\left(\frac{n^{3/2} \sqrt{\log(1/\beta)}}{n-\lambda}\right)$, which yields the desired bound.  
\end{enumerate}

\ifnum\makeappendix=1
    We formalize this intuition in Corollary~\ref{coro:bit-sum-accurate-concrete} \ifnum\oakland=0 to obtain Theorem~\ref{thm:bitsum} in the introduction\fi.  
    
    \medskip
    
    The remainder of the analysis can be found in \ifnum\eurocrypt=0 Appendix \else Supplementary Material \fi \ref{sec:appendix-bitsum}.

\else
    Theorem \ref{thm:bitsum} in the introduction follows from this intuition; a formal proof can be found in the full version.
\fi

\section{A Protocol for Sums of Real Numbers}
\label{sec:real-sum}



In this section, we show how to extend our protocol to compute sums of bounded real numbers.  In this case the data domain is $\cX = [0,1]$, but the function we wish to compute is still $f(x) = \sum_i x_i$.  The main idea of the protocol is to randomly round each number $x_i$ to a Boolean value $b_i \in \zo$ with expected value $x_i$.  However, since the randomized rounding introduces additional error, we may need to round multiple times and estimate several sums. As a consequence, this protocol is not one-message.

\subsection{The Protocol}

Our algorithm is described in two parts, an encoder $E_r$ that performs the randomized rounding (Algorithm~\ref{alg:encode-real}) and a shuffled protocol $P^\R_{\lambda,r}$ (Algorithm~\ref{alg:real-sum-shuffled}) that is the composition of many copies of our protocol for the binary case, $P^{0/1}_{\lambda}$.  The encoder takes a number $x \in [0,1]$ and a parameter $r \in \N$ and outputs a vector $(b_1,\dots,b_r) \in \zo^{r}$ such that $\ex{}{\frac1r \sum_j b_j} = x_j$ and $\var{}{\frac1r \sum_j b_j} = O(1/r^2)$.  To clarify, we give two examples of the encoding procedure:
\begin{itemize}
\item If $r = 1$ then the encoder simply sets $b = \Ber(x)$.  The mean and variance of $b$ are $x$ and $x(1-x) \leq \frac14$, respectively.
\item If $x = .4$ and $r = 4$ then the encoder sets $b = (1,\Ber(.6),0,0)$.  The mean and variance of $\frac14(b_1 + b_2 + b_3 + b_4)$ are $.4$ and $.015$, respectively.  
\end{itemize}

After doing the rounding, we then run the bit-sum protocol $P^{0/1}_{\lambda}$ on the bits $b_{1,j},\dots,b_{n,j}$ for each $j \in [r]$ and average the results to obtain an estimate of the quantity
$$
\sum_{i} \frac{1}{r} \sum_{j} b_{i,j} \approx \sum_{i} x_i
$$
To analyze privacy we use the fact that the protocol is a composition of bit-sum protocols, which are each private, and thus we can analyze privacy via the composition properties of differential privacy.

Much like in the bit-sum protocol, we use $P^\R_{n,\lambda,r}, R^\R_{n,\lambda,r}, A^\R_{n,\lambda,r}$ to denote the real-sum protocol and its components when $n$ users participate.



\begin{algorithm}[ht]
\caption{An encoder $E_{r}(x)$}
\label{alg:encode-real}
\KwIn{$x \in [0,1]$, a parameter $r \in \N$.}
\KwOut{$(\mathbf{b}_{1},\dots, \mathbf{b}_{r}) \in \zo^{r}$} \vspace{5pt}

Let $\mu \leftarrow \lceil x \cdot r \rceil$ and $p \leftarrow x\cdot r - \mu + 1$\\
\For{$j = 1,\dots,r$}{
	$\mathbf{b}_{j} = \begin{cases}1 & j < \mu\\ \mathit{Ber}(p) & j = \mu\\  0 & j > \mu \end{cases}$
}
\Return{$(\mathbf{b}_{1},\dots, \mathbf{b}_{r})$}
\end{algorithm}


\begin{algorithm} [ht]
\caption{The protocol $P^\R_{\lambda,r} = (R^\R_{\lambda,r}, \shuff, A^\R_{\lambda,r})$}
\label{alg:real-sum-shuffled}
\tcp{Local randomizer}
\Fn{$R^\R_{n,\lambda,r}(x)$}{
	\KwIn{$x \in [0,1]$, parameters $n,r \in \N, \lambda\in(0,n)$.}
	\KwOut{$(\mathbf{y}_{1},\dots \mathbf{y}_{r})\in \zo^{r}$} \vspace{5pt}
    
	$(\mathbf{b}_{1},\dots \mathbf{b}_{r})\gets E_{r}(x)$\\
	
    \Return{$(\mathbf{y}_{1},\dots \mathbf{y}_{r})\gets \left( R^{0/1}_{n,\lambda}( \mathbf{b}_{1} ),\dots, R^{0/1}_{n,\lambda}( \mathbf{b}_{r}) \right)$} 
}

\vspace{10pt}

\tcp{Analyzer}
\Fn{$A^\R_{n,\lambda,r}(y_{1,1},\dots, y_{n,r})$}{
	\KwIn{$(y_{1,1},\dots, y_{n,r})\in \zo^{n\cdot r}$, parameters $n,r \in \N, \lambda\in(0,n)$.}
	\KwOut{$z\in [0,n]$} \vspace{5pt}

	\Return{$z\leftarrow \frac{1}{r}\cdot \frac{n}{n-\lambda} \left( \left( \sum_j \sum_i y_{i,j} \right) - \frac{\lambda\cdot r}{2} \right)$} 
}
\end{algorithm}

\begin{thm}
\label{thm:main-real-sum}
For every $\delta = \delta(n)$ such that $e^{-\Omega(n^{1/4})} < \delta(n) < \frac{1}{n}$ and $\frac{\mathrm{poly(\log n)}}{n}$ $< \eps < 1$ and every sufficiently large $n$, there exists parameters $\lambda \in [0,n],r\in \N$ such that $P^\R_{n,\lambda,r}$ is both $(\eps,\delta)$ differentially private and for every $\beta > 0$, and every $X = (x_1,\dots,x_n) \in [0,1]^n$,
$$
\pr{}{\left| P^\R_{n,\lambda,r}(X) - \sum_{i=1}^{n} x_{i} \right| > O\left( \frac{1}{\eps} \log\frac{1}{\delta} \sqrt{ \log\frac{1}{\beta}}\right) } \leq \beta
$$
\end{thm}

\ifnum\eurocrypt=0
\subsection{Warmup: One Message}

To simplify the discussion, we will handle the case where $\eps < 1/\sqrt{n}$ is quite small, in which it suffices to consider $r = 1$.  In this case the protocol is \emph{exactly} the bit-sum protocol run on the bits $r_1,\dots,r_n$.  In this case we have two sources of error, the rounding, and the bit-sum protocol itself, and we can simply analyze the combination.  

The error of the rounding is bounded by Hoeffding's inequality. \jnote{Fix ref to Hoeffding and the broken link.}
\begin{claim}
For every $n \in \N$, $x_1,\dots,x_n \in [0,1]$, and $\beta > 0$,
$
\pr{}{\left|  \sum_{i=1}^{n} x_i - \sum_{i=1}^{n} E_{1}(x_i) \right| > \sqrt{\tfrac{1}{2}n \log(2/\beta)}} \leq \beta.
$
\end{claim}

Using this claim, combined with Corollary~\ref{coro:bit-sum-accurate-concrete}, we immediately obtain the following
\begin{thm}
For every $n \in \N$, $\delta \in (0,1)$, and $\eps \in \left[\frac{\sqrt{3456}}{n} \log \frac{4}{\delta}, \sqrt{\frac1n}\right]$, and every $\beta > 0$, there is a $\lambda$ such that the protocol $P_{n,\lambda,1}$ is $(\eps,\delta)$-differentially private and for every $x_1,\dots,x_n \in \N$,
\begin{mymath}
\pr{}{\left| P_{n,\lambda,1}(x) - \sum_{i=1}^{n} x_i \right| > \alpha} \leq \beta
\end{mymath}
for
$$
\alpha = O\left(\left(\sqrt{n} +  \tfrac{1}{\eps} \sqrt{\log \tfrac{1}{\delta}}\right) \sqrt{\log \tfrac{1}{\beta}}\right) = O\left(\tfrac{1}{\eps} \sqrt{\log \tfrac{1}{\delta} \log \tfrac{1}{\beta}}\right)
$$
\end{thm}

Summing up, when $\eps < 1/\sqrt{n}$, the error term coming from rounding is smaller than the error already in the bit-sum protocol.  Thus, we have established Theorem~\ref{thm:main-real-sum} for the regime where $\eps < 1/\sqrt{n}$.  However, when $\eps$ is larger, the bit-sum protocol has much less than $\sqrt{n}$ error, so we will need to perform the more elaborate rounding with $r > 1$.  
\fi

\subsection{Privacy Analysis}

Privacy will follow immediately from the composition properties of shuffled protocols (Lemma~\ref{ref:mixnetcomposition}) and the privacy of the bit-sum protocol $P_{n,\lambda}$.    One technical nuisance is that the composition properties are naturally stated in terms of $\eps$, whereas the protocol is described in terms of the parameter $\lambda$, and the relationship between $\eps, \lambda,$ and $n$ is somewhat complex.  Thus, we will state our guarantees in terms of the level of privacy that each individual bit-sum protocol achieves with parameter $\lambda$.  To this end, define the function $\lambda^*(n,\eps,\delta)$ to be the minimum value of $\lambda$ such that the bit-sum protocol with $n$ users satisfies $(\eps,\delta)$-differential privacy.  We will state the privacy guarantee in terms of this function.
\begin{thm}
\label{coro:real-sum-private} 
For every $\eps,\delta \in (0,1), n,r \in \N$, define
$$
\eps_0 = \frac{\eps}{\sqrt{8 r \log(2/\delta)}}~~~~~~\delta_0 = \frac{\delta}{2r}~~~~~~\lambda^* = \lambda^*(n,\eps_0,\delta_0)
$$
For every $\lambda \geq \lambda^*$, $P^\R_{n,\lambda,r}$ is $(\eps,\delta)$-differentially private.
\end{thm}

\jnote{Consider giving proof?}

\subsection{Accuracy Analysis}
In this section, we bound the error of $P^\R_{\lambda,r}(X)$ with respect to $\sum_i x_i$. Recall that $f(X) = \sum_i x_i$.


Observe that there are two sources of randomness: the encoding of the input $X=(x_1,\dots x_n)$ as bits and the execution of $R^{0/1}_{n,\lambda}$ on that encoding. We first show that the bit encoding lends itself to an unbiased and concentrated estimator of $f(X)$. Then we show that the output of $P_{n,\lambda,r}$ is concentrated around any value that estimator takes.

\begin{thm}
\label{thm:real-sum-error}
For every $\beta > 0$, $n \geq \lambda \geq \frac{16}{9}\log\frac{2}{\beta}$, $r\in \N$, and $X\in [0,1]^n$,
\[
\pr{}{\left| P^\R_{n,\lambda,r}(X) - f(X) \right| 
\geq \tfrac{\sqrt{2}}{r}\sqrt{ n\log\tfrac{2}{\beta}} + \tfrac{n}{n-\lambda} \cdot \sqrt{2\tfrac{\lambda}{r}  \log \tfrac{2}{\beta}}} < 2\beta
\]
\end{thm}

The analysis can be found in \ifnum\makeappendix=1 \ifnum\eurocrypt=0 Appendix \else Supplementary Material \fi\ref{sec:appendix-realsum}. Later in that section, we argue that setting $r\leftarrow \eps\cdot \sqrt{n}$ suffices to achieve the bound in Theorem \ref{thm:main-real-sum}\else the full version of the paper, which also argues that setting $r\leftarrow \eps\cdot \sqrt{n}$ suffices to achieve the bound in Theorem \ref{thm:main-real-sum}\fi.

\section{Lower Bounds for the Shuffled Model}
\label{sec:lower-bounds}

In this section, we prove separations between central model algorithms and shuffled model protocols where each user's local randomizer is identical and sends one indivisible message to the shuffler (the one-message model).


\begin{thm}[Shuffled-to-Local Transformation]
\label{thm:shuffled-to-local}
Let $P_{S}$ be a protocol in the one-message shuffled model that is
\begin{itemize}
\item $(\eps_{S},\delta_{S})$-differentially private in the shuffled model for some $\eps_{S} \leq 1$ and $\delta_{S} = \delta_{S}(n) < n^{-8}$, and
\item $(\alpha,\beta)$-accurate with respect to $f$ for some $\beta = \Omega(1)$.
\end{itemize}
Then there exists a protocol $P_{L}$ in the local model that is
\begin{itemize}
\item $(\eps_{L}, 0)$-differentially private in  the local model for $\eps_{L} = 8(\eps_{S} + \ln n)$, and
\item $(\alpha,4\beta)$-accurate with respect to $f$ (when $n$ is larger than some absolute constant)
\end{itemize}
\end{thm}

This means that an impossibility result for approximating $f$ in the local model implies a related impossibility result for approximating $f$ in the shuffled model.  In Section~\ref{sec:lbapplications} we combine this result with existing lower bounds for local differential privacy to obtain several strong separations between the central model and the one-message shuffled model.


The key to Theorem~\ref{thm:shuffled-to-local} is to show that if $P_{S} = (R_{S},S,A_{S})$ is a protocol in the one-message shuffled model satisfying $(\eps_S,\delta_S)$-differential privacy, then the algorithm $R_{S}$ itself satisfies $(\eps_{L},\delta_{S})$-differential privacy without use of the shuffler $S$.  Therefore, the local protocol $P_{L} = (R_{S}, A_{S} \circ S)$ is $(\eps_{L},\delta_{S})$-private in the local model and has the exact same output distribution, and thus the exact same accuracy, as $P_{S}$.
To complete the proof, we use (a slight generalization of) a transformation of Bun, Nelson, and Stemmer~\cite{bns18} to turn $R$ into a related algorithm $R'$ satisfying $(8(\eps_{S} + \ln n),0)$-differential privacy with only a slight loss of accuracy. We prove the latter result in \ifnum\makeappendix=1 \ifnum\eurocrypt=0 Appendix \else Supplementary Material \fi ~\ref{sec:approx-to-pure} \else the full version of the paper\fi.

\subsection{One-message Randomizers Satisfy Local Differential Privacy}

The following lemma is the key step in the proof of Theorem~\ref{thm:shuffled-to-local}, and states that for any symmetric shuffled protocol, the local randomizer $R$ must satisfy local differential privacy with weak, but still non-trivial, privacy parameters.


\begin{thm}
\label{thm:remove-shuffler}
Let $P = (R,S,A)$ be a protocol in the one-message shuffled model. If $n \in \N$ is such that $P_{n}$ satisfies $(\eps_{S},\delta_{S})$-differential privacy, then the algorithm $R$ satisfies $(\eps_{L},\delta_L)$-differential privacy for $\eps_{L} = \eps_{S} + \ln n$.  Therefore, the symmetric local protocol $P_{L} = (R,A \circ S)$ satisfies $(\eps_{L},\delta_L)$-differential privacy.
\end{thm}

\begin{proof}

By assumption, $P_n$ is $(\eps_S,\delta_S)$-private.  Let $\eps$ be the supremum such that $R \from \cX \to \cY$ is \emph{not} $(\eps,\delta_{S})$-private.  We will attempt to find a bound on $\eps$.  If $R$ is not $(\eps,\delta_{S})$-differentially private, there exist $Y \subset \univy$ and $x,x'\in \univx$ such that
\[
\pr{}{R(x')\in Y} > \exp(\eps) \cdot \pr{}{R(x)\in Y} +\delta_{S}
\]
For brevity, define $p := \prob(R(x)\in Y)$ and $p' := \prob(R(x')\in Y)$ so that we have
\begin{equation}
p' > \exp(\eps)p +\delta_{S}\label{eq:transform-contradiction}
\end{equation}

We will show that if $\eps$ is too large, then \eqref{eq:transform-contradiction} will imply that $P_n$ is \emph{not} $(\eps_S,\delta_S)$-differentially private, which contradicts our assumption.  To this end, define the set $\cW:=\{W\in \univy^n~|~\exists i~ w_i \in Y\}$. Define two datasets $X\sim X'$ as $$X := (\underbrace{x,\dots,x}_{n \textrm{ times}})~~~\textrm{and}~~~X' := (x',\underbrace{x,\dots,x}_{n-1 \textrm{ times}})$$
Because $P_n$ is $(\eps_S,\delta_S)$-differentially private
\begin{equation} \label{eq:transform-privacy}
\pr{}{P_n(X')\in \cW} \leq \exp(\eps_S)\cdot \pr{}{P_n(X)\in W} + \delta_S 
\end{equation}
Now we have
\begin{align*}
&\pr{}{P_n(X)\in \cW} \\
={} &\pr{}{\shuff(\underbrace{R(x),\dots,R(x)}_{n \textrm{ times}}) \in \cW}  \\
={} &\pr{}{(\underbrace{R(x),\dots,R(x)}_{n \textrm{ times}}) \in \cW} \tag{$\cW$ is symmetric} \\
={} &\pr{}{\exists i~~R(x) \in Y} 
\leq{} n \cdot \pr{}{R(x) \in Y} \tag{Union bound} \\
={} &np
\end{align*}
where the second equality is because the set $W$ is closed under permutation, so we can remove the random permutation $\shuff$ without changing the probability.  Similarly, we have
\begin{align*}
\pr{}{P_n(X')\in \cW}
={} &\pr{}{(R(x'),\underbrace{R(x)\dots,R(x)}_{n-1 \textrm{ times}}) \in \cW}  \\
\geq{} &\pr{}{R(x') \in Y} 
={} p' \\
>{} &\exp(\eps)p + \delta_{S} \tag{By \eqref{eq:transform-contradiction}}
\end{align*}
Now, plugging the previous two inequalities into \eqref{eq:transform-privacy}, we have
\begin{align*}
\exp(\eps)p + \delta_{S} <{} &\pr{}{P_n(X')\in \cW} \\
\leq{} &\exp(\eps_{S})\cdot \pr{}{P_n(X)\in \cW} \\
\leq{} &\exp(\eps_{S})np + \delta_{S}
\end{align*}
By rearranging and canceling terms in the above we obtain the conclusion
$$
\eps \leq \eps_{S} + \ln n
$$
Therefore $R$ must satisfy $(\eps_{S} + \ln n, \delta_{S})$-differential privacy.
\end{proof}

\begin{clm}
\label{clm:sameAnalysis}
If the shuffled protocol $P_{S} = (R,S,A)$ is $(\alpha,\beta)$-accurate for some function $f$, then the local protocol $P_{L} = (R, A \circ S)$ is $(\alpha,\beta)$-accurate for $f$, where
$$
(A \circ S)(y_1,\dots,y_N) = A( S(y_1,\dots,y_N))
$$
\end{clm}

We do not present a proof of Claim \ref{clm:sameAnalysis}, as it is immediate that the distribution of $P_{S}(x)$ and $P_{L}(x)$ are identical, since $A \circ S$ incorporates the shuffler.

We conclude this section with a slight extension of a result of Bun, Nelson, and Stemmer~\cite{bns18} showing how to transform any local algorithm satisfying $(\eps,\delta)$-differential privacy into one satisfying $(O(\eps),0)$-differential privacy with only a small decrease in accuracy.  Our extension covers the case where $\eps > 2/3$, whereas their result as stated requires $\eps \leq 1/4$.

\begin{thm} [Extension of~\cite{bns18}]
\label{thm:transform}
Suppose local protocol $P_{L} = (R,A)$ is $(\eps,\delta)$ differentially private and $(\alpha,\beta)$ accurate with respect to $f$. If $\eps > 2/3$ and 
$$
\delta < \frac{\beta}{8n\ln(n/\beta)}\cdot \frac{1}{\exp(6\eps)}
$$
then there exists another local protocol $P_{L}' = (R',A)$ that is $(8\eps,0)$ differentially private and $(\alpha,4\beta)$ accurate with respect to $f$.
\end{thm}

The proof can be found in \ifnum\makeappendix=1\ifnum\eurocrypt=0 Appendix \else Supplementary Material \fi ~\ref{sec:approx-to-pure} \else the full version of the paper\fi. Theorem~\ref{thm:shuffled-to-local} now follows by combining Theorem~\ref{thm:remove-shuffler} and Claim~\ref{clm:sameAnalysis} with Theorem~\ref{thm:transform}.

\subsection{Applications of Theorem \ref{thm:shuffled-to-local}} \label{sec:lbapplications}

In this section, we define two problems and present known lower bounds in the central and local models. By applying Theorem \ref{thm:shuffled-to-local}, we derive lower bounds in the one-message shuffled model. These bounds imply large separations between the central and one-message shuffled models.

\renewcommand{\arraystretch}{1.25}
\begin{table}[t]
\centering
\caption{Comparisons Between Models.  When a parameter is unspecified, the reader may substitute $\eps = 1, \delta = 0, \alpha = \beta = .01$.  \textbf{All results are presented as the minimum dataset size $n$ for which we can hope to achieve the desired privacy and accuracy as a function of the relevant parameter for the problem.}}

\label{table:model-comparisons}

\begin{tabular}{|c|c|c|c|c|}
\hline
\multirow{3}{*}{\begin{tabular}[c]{@{}c@{}}Function\\ \\ (Parameters)\end{tabular}} & \multicolumn{4}{c|}{Differential Privacy Model} \\ \cline{2-5} 
 & \multirow{2}{*}{Central} & \multicolumn{2}{c|}{\textbf{Shuffled (this paper)}} & \multirow{2}{*}{Local} \\ \cline{3-4}
 &  & \textbf{One-Message} & \textbf{General} &  \\ \hline
\begin{tabular}[c]{@{}c@{}}Mean, $\mathcal{X} = \{0,1\}$\\ (Accuracy $\alpha$)\end{tabular} & \multirow{2}{*}{$\Theta(\frac{1}{\alpha \eps})$}                 & \multicolumn{2}{c|}{$O\big(\frac{\sqrt{\log(1/\delta)}}{\alpha \eps}\big)$}                                                                                         & \multirow{2}{*}{$\Theta(\frac{1}{\alpha^2 \eps^2})$} \\ \cline{1-1} \cline{3-4}
\begin{tabular}[c]{@{}c@{}}Mean, $\mathcal{X} = [0,1]$\\ (Accuracy $\alpha$)\end{tabular}   &                                                                  & $O\big(\frac{1}{\alpha^2} + \frac{\sqrt{\log(1/\delta)}}{\alpha \eps}\big)$        & $O\big(\frac{\log(1/\delta)}{\alpha \eps}\big)$                                &                                                      \\ \hline
\begin{tabular}[c]{@{}c@{}}Selection\\ (Dimension $d$)\end{tabular}                         & $\Theta(\log d)$                                                 & $\Omega(d^{\frac{1}{17}})$                                                         & $\tilde{O}(\sqrt{d } \log \frac{d}{\delta})$                                     & $\Theta(d \log d)$                                   \\ \hline
\begin{tabular}[c]{@{}c@{}}Histograms\\ (Domain Size $D$)\end{tabular}                      & $\Theta\big(\min\big\{ \log\frac{1}{\delta}, \log D \big\}\big)$ & $\Omega(\log^{\frac{1}{17}} D)$                                                    & $O(\sqrt{\log D})$                                                             & $\Theta(\log D)$                                     \\ \hline
\end{tabular}
\end{table}

\subsubsection{The Selection Problem}

We define the \emph{selection problem} as follows.  The data universe is $\cX = \zo^{d}$ where $d$ is the \emph{dimension} of the problem and the main parameter of interest.  Given a dataset $x = (x_1,\dots,x_n) \in \cX^n$, the goal is to identify a coordinate $j$ such that the sum of the users' $j$-th bits is approximately as large as possible.  That is, a coordinate $j \in [d]$ such that
\begin{equation} \label{eq:selection}
\sum_{i=1}^{n} x_{i,j} \geq \max_{j' \in [d]} \sum_{i=1}^{n} x_{i,j'} - \frac{n}{10}
\end{equation}
We say that an algorithm \emph{solves the selection problem with probability $1-\beta$} if for every dataset $x$, with probability at least $1-\beta$, it outputs $j$ satisfying~\eqref{eq:selection}.

We would like to understand the minimum $n$ (as a function of $d$) such that there is a differentially private algorithm that can solve the selection problem with constant probability of failure.  We remark that this is a very weak notion of accuracy, but since we are proving a negative result, using a weak notion of accuracy only strengthens our results.  

The following lower bound for locally differentially private protocols for selection is from~\cite{selection}, and is implicit in the work of~\cite{DuchiJW13}.\footnote{These works assume that the dataset $x$ consists of independent samples from some distribution $\mathcal{D}$, and define accuracy for selection with respect to mean of that distribution.  By standard arguments, a lower bound for the distributional version implies a lower bound for the version we have defined.}
\begin{thm} If $P_{L} = (R_{L},A_{L})$ is a local protocol that satisfies $(\eps,0)$-differen\-tial privacy and $P_{L}$ solves the selection problem with probability $\frac{9}{10}$ for datasets $x \in (\zo^{d})^n$, then $n = \Omega\left(\frac{d \log d}{(e^{\eps} - 1)^2}\right)$.
\end{thm}

By applying Theorem~\ref{thm:shuffled-to-local} we immediately obtain the following corollary.
\begin{coro}
If $P_{S} = (R_{S},S,A_{S})$ is a $(1,\delta)$-differentially private protocol in the one-message shuffled model, for $\delta = \delta(n) < n^{-8}$, and $P_{S}$ solves the selection problem with probability $\frac{99}{100}$, then $n = \Omega( (d \log d)^{1/17} )$.
\end{coro}

Using a multi-message shuffled protocol\footnote{The idea is to simulate multiple rounds of our protocol for binary sums, one round per dimension.}, we can solve selection with $\tilde{O}(\frac{1}{\eps} \sqrt{d})$ samples.  By contrast, in the local model $n = \Theta(\frac{1}{\eps^2} d \log d)$ samples are necessary and sufficient.  In the central model, this problem is solved by the \emph{exponential mechanism}~\cite{McSherryT07} with a dataset of size just $n = O(\frac{1}{\eps} \log d)$, and this is optimal~\cite{BafnaU17,SteinkeU17}. These results are summarized in Table~\ref{table:model-comparisons}.

\subsubsection{Histograms}
We define the \emph{histogram problem} as follows.  The data universe is $\cX = [D]$ where $D$ is the \emph{domain size} of the problem and the main parameter of interest.  Given a dataset $x = (x_1,\dots,x_n) \in \cX^n$, the goal is to build a vector of size $D$ such that for all $j\in[D]$ the $j$-th element is as close to the frequency of $j$ in $x$. That is, a vector $v \in [0,n]^D$ such that
\begin{equation}
\label{eq:histogram}
\max_{j\in[D]} \left| v_j - \sum_{i=1}^n \mathds{1}(x_i=j) \right| \leq \frac{n}{10}
\end{equation}
where $\mathds{1}(\texttt{conditional})$ is defined to be $1$ if \texttt{conditional} evaluates to \texttt{true} and $0$ otherwise.

Similar to the selection problem, an algorithm \emph{solves the histogram problem with probability} $1-\beta$ if for every dataset $x$, with probability at least $1-\beta$ it outputs $v$ satisfying \eqref{eq:histogram}. We would like to find the minimum $n$ such that a differentially private algorithm can solve the histogram problem; the following lower bound for locally differentially private protocols for histograms is from \cite{BassilyS15}.

\begin{thm}
If $P_L=(R_L,A_L)$ is a local protocol that satisfies $(\eps,0)$ differential privacy and $P_L$ solves the histogram problem with probability $\frac{9}{10}$ for any $x\in [D]^n$ then $n=\Omega\left(\frac{\log D}{(e^\eps-1)^2}\right)$
\end{thm}

By applying Theorem \ref{thm:shuffled-to-local}, we immediately obtain the following corollary.
\begin{coro}
If $P_{S} = (R_{S},S,A_{S})$ is a $(1,\delta)$-differentially private protocol in the one-message shuffled model, for $\delta = \delta(n) < n^{-8}$, and $P_{S}$ solves the histogram problem with probability $\frac{99}{100}$, then $n=\Omega\left(\log^{1/17}D\right)$
\end{coro}

In the shuffled model, we can solve this problem using our protocol for bit-sums by having each user encode their data as a ``histogram'' of just their value $x_i \in [D]$ and then running the bit-sum protocol $D$ times, once for each value $j \in [D]$, which incurs error $O(\frac{1}{\eps} \sqrt{ \log \frac{1}{\delta} \log D})$.\footnote{Note that changing one user's data can only change two entries of their local histogram, so we only have to scale $\eps,\delta$ by a factor of $2$ rather than a factor that grows with $D$.} But in the central model, this problem can be solved to error $O( \min \{ \log \frac{1}{\delta}, \log D \})$, which is optimal (see, e.g.~\cite{Vadhan16}). Thus, the central and one-message shuffled models are qualitatively different with respect to computing histograms: $D$ may be infinite in the former whereas $D$ must be bounded in the latter.



\section*{Acknowledgements}
AC was supported by NSF award CCF-1718088.  AS was supported by NSF awards IIS-1447700 and AF-1763786 and a Sloan Foundation Research Award.  JU was supported by NSF awards CCF-1718088, CCF-1750640, CNS-1816028 and a Google Faculty Research Award. 

\ifnum\oakland=1
     \footnotesize
     \bibliographystyle{abbrvnat}
\else
	\ifnum\eurocrypt=1
		\bibliographystyle{splncs04}
	\else
		\bibliographystyle{abbrvnat}
	\fi
\fi
\bibliography{mixnetmodel,local,extrareferences}

\ifnum\oakland=1
\end{document}
\fi

\ifnum\makeappendix=1

\ifnum\eurocrypt=1
\newpage
\begin{center}
  \Large \bf Supplementary Material
\end{center}
\fi

\appendix

\section{Privacy of Bit Sum Protocol}\label{sec:appendix-bitsum-privacy}
Here, we prove Claim \ref{clm:C_lambda}, which expresses the privacy of $C_\lambda$ in terms of $\lambda$:

\begin{clm*}[Restatement of Claim \ref{clm:C_lambda}]
For any $\delta>0$ and $n \geq \lambda \geq 14\log\frac{4}{\delta}$, $C_{\lambda}$ is $(\eps, \delta )$ differentially private where
\[
\eps = \sqrt{\frac{32\log\frac{4}{\delta}}{\lambda- \sqrt{2\lambda\log\tfrac{2}{\delta}}}} \cdot \left( 1 - \frac{\lambda- \sqrt{2 \lambda\log\tfrac{2}{\delta}}}{n} \right)
\]
\end{clm*}

\begin{proof}
Fix any $X\sim X'\in\zo^n$ and any $W \subseteq [n]$.
\begin{align*}
\pr{}{C_{\lambda}(X) \in W} &= \pr{}{C_{\lambda}(X) \in W \cap
                              \mathbf{s}\geq \lambda-\sqrt{2
                              \lambda\log\tfrac{2}{\delta}}} \\
  & \quad \quad \quad \quad+ \pr{}{C_{\lambda}(X) \in W \cap \mathbf{s}<\lambda-\sqrt{2 \lambda\log\tfrac{2}{\delta}} }\\
&\leq \pr{}{C_{\lambda}(X) \in W \cap \mathbf{s}\geq \lambda-\sqrt{2 \lambda\log\tfrac{2}{\delta}}} + \half\delta \tag{Chernoff bound}\\
&= \sum_{s\geq \lambda-\sqrt{2 \lambda\log\frac{2}{\delta}} } \pr{}{C_s(X) \in W}\cdot \pr{}{\mathbf{s}=s} +\half\delta \stepcounter{equation}\tag{\theequation}\label{eq:C_lambda-privacy}
\end{align*}
Because $\lambda$ is sufficiently large, $\lambda-\sqrt{2 \lambda\log\frac{2}{\delta}} > 8\log\frac{4}{\delta}$. Claim \ref{clm:C_s} thus applies to each term in the sum.

\[
\pr{}{C_s(X) \in W} \leq \exp\left( \sqrt{\frac{32\log\frac{4}{\delta}}{s}} \cdot \left(1-\frac{s}{n}\right)\right)\cdot \pr{}{C_s(X') \in W} + \half \delta
\]

For notational convenience, we will use the shorthand $\eps_1(s) := \sqrt{\frac{32\log\frac{4}{\delta}}{s}} \cdot \left(1-\frac{s}{n}\right)$. Therefore,
\begin{align*}
\eqref{eq:C_lambda-privacy} &\leq \left( \sum_{s\geq \lambda-\sqrt{2 \lambda\log\tfrac{2}{\delta}} } \left[ e^{\eps_1(s)} \cdot\pr{}{C_s(X') \in W} + \half\delta \right] \cdot \pr{}{\mathbf{s}=s}\right) +\half\delta\\
&\leq \left( \sum_{s\geq \lambda-\sqrt{2 \lambda\log\tfrac{2}{\delta}} } e^{\eps_1(s)} \cdot\pr{}{C_s(X') \in W} \cdot \pr{}{\mathbf{s}=s} \right) +\delta\\
&\leq \max_{s\geq \lambda-\sqrt{2 \lambda\log\tfrac{2}{\delta}} } e^{\eps_1(s)} \cdot \pr{}{C_{\lambda}(X') \in W} +\delta
\end{align*}
Because $\eps_1(s)$ decreases with $s$, the above is maximized at the lower bound on $s$:
\begin{align*}
\pr{}{C_{\lambda}(X) \in W} \leq& \exp\left(\sqrt{\frac{32\log\frac{4}{\delta}}{\lambda- \sqrt{2\lambda\log\tfrac{2}{\delta}}}} \cdot \left( 1 - \frac{\lambda- \sqrt{2 \lambda\log\tfrac{2}{\delta}}}{n} \right) \right) \cdot \\
&~~~~~ \pr{}{C_{\lambda}(X') \in W} + \delta
\end{align*}
which completes the proof.
\end{proof}

\section{Accuracy of Bit Sum Protocol}\label{sec:appendix-bitsum}
In this section, we will show there exists a value of $\lambda$ for any (sufficiently large) target value of $\eps$ such that the error of bit sum protocol $P_{n,\lambda}$ is $\tilde{O}((1/\eps)\sqrt{\log (1/\delta)})$. But first, we shall bound the error of the protocol in terms of $\lambda$:

\begin{thm*}[Restatement of Theorem \ref{thm:bit-sum-accurate}]
For every $n \in \N$, $\beta > 0$, $n > \lambda \geq 2\log\frac{2}{\beta}$, and $x \in \zo^{n}$,
$$
\pr{}{\left| P_{n,\lambda}(x) - \sum_i x_i \right| > \sqrt{2\lambda \log(2/\beta)} \cdot \left(\frac{n}{n-\lambda}\right)} \leq \beta
$$
\end{thm*}

\medskip

Recall that each user $i$ sends a message $y_i$ which is a randomization of bit $x_i$. We begin the analysis by determining the mean and variance of each $y_i$:
\begin{clm}
\label{clm:bit-mean-variance}
For any $n\in \N$, $0 < \lambda \leq n$, and $x \in \bits$,
\begin{align*}
\ex{}{R_{n,\lambda}(x)} &= \frac{\lambda}{2n} + \left(1-\frac{\lambda}{n}\right)\cdot x \\
\var{}{R_{n,\lambda}(x)} &= \frac{\lambda}{2n}\cdot \left(1 - \frac{\lambda}{2n}\right) \end{align*}
\end{clm}
\begin{proof} For shorthand, $\mathbf{y} = R_{n,\lambda}(x)$. The calculation of the expectation is not long:
\begin{align*}
\ex{}{\mathbf{y}} &= \frac{\lambda}{n}\cdot \ex{}{\mathit{Ber}\left(\half\right)} + \left(1-\frac{\lambda}{n}\right)\cdot x\\
&= \frac{\lambda}{2n} + \left(1-\frac{\lambda}{n}\right)\cdot x
\end{align*}

If $x=0$, then $\mathbf{y}$ is a Bernoulli random variable with probability $\half\cdot \frac{\lambda}{n}$ of being 1. The variance of $\mathit{Ber}(p)$ is $p(1-p)$, which is here $\frac{\lambda}{2n}\left(1 - \frac{\lambda}{2n}\right)$. A symmetric argument applies to the case where $x=1$. This concludes the proof.
\end{proof}

Using the above, one can compute the mean and variance of the protocol's estimate using linearity of expectation and the fact that the output of the protocol is a (rescaled) sum of independent messages:
\begin{clm}
For any $n\in \N$, $0 < \lambda \leq n$, and $X \in \bits^n$,
\begin{align*}
\ex{}{P_{n,\lambda}(X)} &= \sum_{i=1}^{n} x_i\\
\var{}{P_{n,\lambda}(X)} &= \left( \frac{n}{n-\lambda} \right)^2 \cdot \frac{\lambda}{2} \cdot \left(1 - \frac{\lambda}{2n}\right)
\end{align*}
\end{clm}




We omit the proof for space.  From the previous claim, and the fact that the protocol is output a sum of independent bits, we can obtain a high-probability bound on the error.  \ifnum\oakland=1 The proof is omitted for space. \fi 

\begin{coro}
For any $n\in \N$, $0<\beta<1$, and $\frac{16}{9} \log\frac{2}{\beta} < \lambda < n$, the protocol $P_{n,\lambda}$ is $(\alpha,\beta)$-accurate for \begin{mymath}
\alpha = \frac{n}{n-\lambda}\sqrt{2\lambda \log\frac{2}{\beta}}
\end{mymath}
\end{coro}

\ifnum\oakland=0
\begin{proof}
Fix any $X\in \zo^n$. Let $\mathbf{d_i}$ denote the random variable $R_{n,\lambda}(x_i) - \frac{\lambda}{2n} - \left(1-\frac{\lambda}{n}\right)\cdot x_i$. It has maximum $1-\frac{\lambda}{2n} < 1$ and minimum $-1+\frac{\lambda}{2n} > -1$. From Claim \ref{clm:bit-mean-variance}, $\ex{}{\mathbf{d_i}}=0$ and $\var{}{\mathbf{d_i}}=\frac{\lambda}{2n}\left(1 - \frac{\lambda}{2n}\right)$. Because $\lambda$ is sufficiently large, the variance is larger than $\frac{4}{9n}\log\frac{2}{\beta}$. By assumption, there are $n$ honest users. These facts imply that Bernstein's inequality (Theorem \ref{thm:bernstein}) is compatible:
\begin{equation}
\label{eq:bernstein-1}
\pr{}{\left| \sum_{i=1}^n \mathbf{d_i} \right| > \sqrt{ 2\lambda\left(1-\frac{\lambda}{2n}\right) \log\frac{2}{\beta}}} < \beta
\end{equation}

Define $\mathbf{y}_i := R_{\lambda}(x_i)$ for shorthand. Observe that 
\begin{align*}
\sum_i \mathbf{d_i} &= \sum_i \left(\mathbf{y_i} - \frac{\lambda}{2n} - \left(1-\frac{\lambda}{n}\right)\cdot x_i \right) \\
&= \left(\sum_i \mathbf{y_i} \right) - \frac{\lambda}{2} - \left(1-\frac{\lambda}{n}\right)\cdot f(x)\\
\frac{n}{n-\lambda} \sum_i \mathbf{d_i} &= \frac{n}{n-\lambda} \left(\left(\sum_i \mathbf{y_i} \right) - \frac{\lambda}{2} \right) - f(x)\\
&= (A_{\lambda}\circ \shuff)(\mathbf{y}_1,\dots,\mathbf{y}_n) - f(x) \tag{$\shuff$ only permutes}\\
&= P_{n,\lambda}(x) - f(x) \stepcounter{equation} \tag{\theequation} \label{eq:bernstein-2}
\end{align*}

Substitution of \eqref{eq:bernstein-2} in \eqref{eq:bernstein-1} yields
\[
\pr{}{\left| P_{n,\lambda}(x) - f(x) \right| > \frac{n}{n-\lambda}\sqrt{ 2\lambda\left(1-\frac{\lambda}{2n}\right) \log\frac{2}{\beta}}} < \beta
\]
The Claim follows from the fact that $\lambda > 0$. This concludes the proof.

\end{proof}
\fi

When $\lambda$ is set to the piecewise function in Lemma \ref{lem:generate-lambda}, the error of $P_{\lambda}$ with respect to the bit-sum is of the same order as the Gaussian mechanism:

\begin{coro}
\label{coro:bit-sum-accurate-concrete}
For any $n \in \N$, $0 < \delta < 1$, $\frac{\sqrt{3456}}{n}\log\frac{4}{\delta} < \eps < 1$, and $\delta < \beta < 1$, there exists a $\lambda \in [0,n]$ such that $P_{n,\lambda}$ is $(\eps,\delta)$-differentially private and for every $X \in \zo^{n}$, with probability at least $1-\beta$,
\[
\left| P_{n,\lambda}(X) - \sum_{i=1}^{n} x_i \right| \leq \frac{30}{\eps}\sqrt{ \log\frac{2}{\beta} \log\frac{4}{\delta}}
\]
\end{coro}

\begin{proof}
Fix any $X\in \zo^n$. Let 
\[
\lambda =\begin{cases}
\frac{64}{\eps^2}\log\frac{4}{\delta} & \textrm{if $\eps \geq \sqrt{\frac{192}{n}\log\frac{4}{\delta}}$}\\
n-\frac{\eps n^{3/2}}{\sqrt{432\log (4/\delta)}} & \textrm{otherwise}
\end{cases}
\]
By Lemma~\ref{lem:generate-lambda}, $P_{\lambda}$ is $(\eps,\delta)$ private. By Theorem \ref{thm:bit-sum-accurate}, $P_{n,\lambda}$ is $(\alpha,\beta)$ accurate for
\[
\alpha = \sqrt{2\lambda \log\frac{2}{\beta}} \cdot \frac{n}{n-\lambda}
\]

\noindent We have to consider both ranges of $\eps$.  If $\eps < \sqrt{\frac{192}{n}\log\frac{4}{\delta}}$, then 
\begin{align*}
\alpha &= \sqrt{2\lambda \cdot \log\frac{2}{\beta}} \cdot \frac{n}{n-\left( n - \frac{\eps n^{3/2}}{\sqrt{432\log(4/\delta)}}\right)} \\
&= \sqrt{2\lambda \cdot \log\frac{2}{\beta}} \cdot \frac{\sqrt{432\log(4/\delta)}} {\eps \sqrt{n}} \\
&< \frac{30}{\eps} \sqrt{\frac{\lambda}{n} \cdot \log\left(\frac{2}{\beta}\right) \log\frac{4}{\delta}}\\
&\leq \frac{30}{\eps} \sqrt{\log\left(\frac{2}{\beta}\right) \log\frac{4}{\delta}} \tag{$\lambda \leq n$}
\end{align*}

\noindent If $\eps \geq \sqrt{\frac{192}{n}\log\frac{4}{\delta}}$, then 
\begin{align*}
\alpha &= \sqrt{2\cdot \frac{64}{\eps^2}\log\left(\frac{4}{\delta}\right) \cdot \log\frac{2}{\beta}} \cdot \frac{n}{n-\lambda} \\&\leq \frac{1}{\eps}\sqrt{128\log\left(\frac{4}{\delta}\right) \cdot \log\frac{2}{\beta}} \cdot \frac{3}{2} \tag{$\lambda < n/3$}\\
&< \frac{17}{\eps} \sqrt{\log\left(\frac{4}{\delta}\right) \log\frac{2}{\beta}}
\end{align*}

\noindent Combining the two bounds completes the proof.
\end{proof}


\section{Accuracy of Real Sum Protocol} \label{sec:appendix-realsum}
As in the previous section, we will show that there exists values of parameters $\lambda,r$ where the error of the real sum protocol $P^\R_{n,\lambda,r}$ is $\tilde{O}((1/\eps)\log (1/\delta))$. But we begin by bounding the error of the protocol in terms of $\lambda, r$:

\begin{thm*}[Restatement of Theorem \ref{thm:real-sum-error}]
For every $\beta > 0$, $n \geq \lambda \geq \frac{16}{9}\log\frac{2}{\beta}$, $r\in \N$, and $X\in [0,1]^n$,
\[
\pr{}{\left| P^\R_{n,\lambda,r}(X) - f(X) \right| 
\geq \tfrac{\sqrt{2}}{r}\sqrt{ n\log\tfrac{2}{\beta}} + \tfrac{n}{n-\lambda} \cdot \sqrt{2\tfrac{\lambda}{r}  \log \tfrac{2}{\beta}}} < 2\beta
\]
\end{thm*}

\medskip
It will help to establish notation. Throughout this section, we fix some $X=(x_1,\dots, x_n)\in \zo^n$. For any $i\in [n]$, the vector of bits $(\mathbf{b}_{i,1},\dots \mathbf{b}_{i,r})$ denotes randomized encoding from $E_r(x_i)$. The set of all such bits $\mathbf{b}_{1,1},\dots \mathbf{b}_{n,r}$ is denoted $\mathbf{B} \in \zo^{n\cdot r}$.

\begin{clm}
\label{clm:real-sum-encoding-error}
For every $n,r\in N$, $X\in (0,1)^n$, and $0 < \beta < 1$,

\begin{equation}
\label{eq:real-sum-encoding-error}
\pr{\mathbf{B}}{\left| f(X) - \frac{1}{r} \sum_{i=1}^n \sum_{j=1}^r \mathbf{b}_{i,j}  \right| > \frac{\sqrt{2}}{r}\sqrt{n\log\frac{2}{\beta}} } \leq \beta
\end{equation}
\end{clm}

\begin{proof}
Fix some $i\in [n]$. Only one bit among $b_{i,1},\dots,b_{i,r}$ is random, at index $\mu$. The remainder have sum $\lfloor x_i\cdot r\rfloor$. Hence,
\begin{equation}
\label{eq:unary-small-error}
\left| x_i - \frac{1}{r} \sum_{j=1}^r \mathbf{b}_{i,j} \right| \leq \frac{1}{r}
\end{equation}

\ifnum\oakland=1
One can verify from our definition of $\mu_i,p_i$ that $\frac{1}{r}(\mathbf{b}_{i,1}+ \dots + \mathbf{b}_{i,r})$ has expected value $x_i$
\else
We show that the average of $\frac{1}{r}(\mathbf{b}_{i,1}+ \dots + \mathbf{b}_{i,r})$ has expected value $x_i$:
\begin{align*}
\ex{\mathbf{b}_{i,1},\dots \mathbf{b}_{i,r}}{\frac{1}{r} \sum_{j=1}^r \mathbf{b}_{i,j}  } &= \ex{}{\frac{1}{r} \left(\mathbf{b}_{i,\mu_i} + \sum_{j\neq \mu} b_{i,j} \right)}\\
&= \frac{1}{r} \left(\ex{}{ \mathbf{b}_{i,\mu} }+ \mu_i-1 \right)\\
&= \frac{1}{r} (p_i + \mu-1)\\
&= \frac{1}{r} (x_i\cdot r - \mu + 1 + \mu-1)\\
&= \frac{1}{r} \cdot x_i\cdot r\\
&= x_i \stepcounter{equation}\label{eq:unary-centered}\tag{\theequation}
\end{align*}
\fi

From \eqref{eq:unary-small-error}, \ifnum\oakland=0\eqref{eq:unary-centered}, and Hoeffding's inequality (Theorem \ref{thm:hoeffding}),\fi the sum over all $x_i - \frac{1}{r} \sum_{j=1}^r \mathbf{b}_{i,j}$ is concentrated:
\[
\pr{\mathbf{B}}{\left|\sum_{i=1}^n x_i - \frac{1}{r}\sum_{i=1}^n \sum_{j=1}^r \mathbf{b}_{i,j} \right| > \frac{2}{r}\sqrt{\half n\log\frac{2}{\beta}} } \leq \beta
\]
which is equivalent to \eqref{eq:real-sum-encoding-error}. This concludes the proof.
\end{proof}

Condition on an encoding $b_{1,1},\dots,b_{n,r}$ of the $n$ real-valued inputs. When we treat the output of the protocol as an estimator of $\frac{1}{r}\sum_j\sum_i b_{i,j}$, we find that it is unbiased and concentrated:

\begin{clm}
\label{clm:real-sum-protocol-error}
For every $\beta > 0$, $n \geq \lambda \geq \frac{16}{9}\log\frac{2}{\beta}$, $X\in (0,1)^n$, $r\in \N$ and every \emph{fixed} set of bits $B=(b_{1,1},\dots,b_{n,r}) \in \bits^{n\cdot r}$,
\ifnum\oakland=1
if we condition on those bits being chosen by the encoder, then
\begin{equation}
\label{eq:real-sum-protocol-error}
\pr{}{ \left| P^\R_{n,\lambda,r}(X) - \frac{1}{r}\sum_j\sum_i b_{i,j} \right| > \frac{n}{n-\lambda}\sqrt{2 \frac{\lambda}{r} \log\frac{2}{\beta}}} < \beta
\end{equation}
\else
\begin{equation}
\label{eq:real-sum-protocol-error}
\pr{}{ \left| P^\R_{n,\lambda,r}(X) - \frac{1}{r}\sum_j\sum_i b_{i,j} \right| > \frac{n}{n-\lambda}\sqrt{2 \frac{\lambda}{r} \log\frac{2}{\beta}} ~\left\vert\vphantom{\frac{1}{1}}\right.~ \mathbf{B} = B } < \beta
\end{equation}
\fi
\end{clm}

\begin{proof}
As in the statement, fix any $B= (b_{1,1},\dots, b_{n,r}) \in \zo^{n\cdot r}$. Let $\mathbf{d}_{i,j}$ denote $R_{n,\lambda}(b_{i,j}) - \frac{\lambda}{2n} - \left(1-\frac{\lambda}{n}\right) \cdot b_{i,j}$. Its value is in $[-1,1]$.  Applying Claim \ref{clm:bit-mean-variance} here, we have $\ex{}{\mathbf{d}_{i,j}}=0$ and $\var{}{\mathbf{d}_{i,j}}=\frac{\lambda}{2n}\left(1- \frac{\lambda}{2n} \right)$. From Bernstein's inequality\ifnum\oakland=0 (Theorem \ref{thm:bernstein})\fi,



\begin{equation}
\pr{}{\left| \sum_j \sum_i \mathbf{d}_{i,j} \right| > \sqrt{2\lambda r\left(1-\frac{\lambda}{2n}\right)\log\frac{2}{\beta}}} < \beta \label{eq:bernstein-real-1}
\end{equation}

Let $\mathbf{y_{i,j}}$ denote the random variable output by $R_\lambda(b_{i,j})$. 
\ifnum\oakland=0
Observe that
\begin{align*}
\sum_j \sum_i \mathbf{d}_{i,j} &= \left(\sum_j \sum_i \mathbf{y}_{i,j} \right) - \frac{\lambda\cdot r}{2} - \left(1- \frac{\lambda}{n}\right) \sum_j \sum_i b_{i,j}\\
\frac{1}{r}\cdot \frac{n}{n-\lambda} \sum_j \sum_i \mathbf{d}_{i,j} &= \frac{1}{r}\cdot \frac{n}{n-\lambda} \left( \left( \sum_j \sum_i \mathbf{y_{i,j}} \right) - \frac{\lambda\cdot r}{2} \right) - \frac{1}{r} \sum_j \sum_i b_{i,j}\\
&= A^\R_{n,\lambda,r}(\mathbf{y}_{1,1},\dots, \mathbf{y}_{n,r})- \frac{1}{r}\sum_j \sum_i b_{i,j} \tag{Defn. of $A^\R_{n,\lambda,r}$}\\
&= (A^\R_{n,\lambda,r} \circ \shuff)(\mathbf{y}_{1,1},\dots, \mathbf{y}_{n,r}) - \frac{1}{r}\sum_j \sum_i b_{i,j} \tag{$\shuff$ only permutes}\\
&= (A^\R_{n,\lambda,r} \circ \shuff \circ R^{0/1}_{n,\lambda})(b_{1,1},\dots, b_{n,r}) - \frac{1}{r}\sum_j \sum_i b_{i,j} \stepcounter{equation} \tag{\theequation} \label{eq:bernstein-real-2}
\end{align*}
\else
A calculation shows that
\begin{align*}
&\sum_j \sum_i \mathbf{d}_{i,j} \\ 
={} &(A^\R_{n,\lambda,r} \circ \shuff \circ R^{0/1}_{n,\lambda})(b_{1,1},\dots, b_{n,r}) - \frac{1}{r}\sum_j \sum_i b_{i,j} \stepcounter{equation} \tag{\theequation} \label{eq:bernstein-real-2}
\end{align*}
\fi

When executing $P^\R_{n,\lambda,r}$, condition on $\mathbf{B} = B$. By substituting \eqref{eq:bernstein-real-2} into \eqref{eq:bernstein-real-1}, we have that
\ifnum\oakland=0
\[
\pr{}{\left| P^\R_{n,\lambda,r}(x) - \frac{1}{r}\sum_j\sum_i b_{i,j} \right| > \frac{n}{n-\lambda}\sqrt{2 \frac{\lambda}{r} \left(\frac{2n-\lambda}{2n}\right)\log\frac{2}{\beta}} \left\vert\vphantom{\frac{1}{1}}\right. \mathbf{B} = B } < \beta
\]
\else
\[
\pr{}{\left| P^\R_{n,\lambda,r}(x) - \tfrac{1}{r}\textstyle\sum_{i,j}  b_{i,j} \right| > \tfrac{n}{n-\lambda}\sqrt{2 \tfrac{\lambda}{r} \left(1-\tfrac{\lambda}{2n}\right)\log\tfrac{2}{\beta}}  } < \beta
\]
\fi
Now, \eqref{eq:real-sum-protocol-error} follows from the fact that $\lambda > 0$. This concludes the proof.
\end{proof}

Theorem \ref{thm:real-sum-error} now follows from a union bound over \eqref{eq:real-sum-encoding-error} and \eqref{eq:real-sum-protocol-error}.

We have bounded the error of $P_{n,\lambda,r}$ in terms of $\lambda,r$. When $\lambda,r$ are chosen such that $P_{n,\lambda,r}$ satisfies $(\eps,\delta)$ differential privacy, we can bound the error in terms of $\eps,\delta$. Theorem \ref{thm:main-real-sum} follows from the following statement:

\begin{coro}
\label{coro:real-sum-accurate}
For any $n > 10^4$, $\delta \in \left( 8 e^{-0.03n^{1/4}}, 1/n\right)$, $\beta \in (\delta, 1)$, $\eps \in \left[ \frac{122}{n} \log\frac{8}{\delta} \sqrt{\log\frac{2}{\beta}}, 1 \right]$, there exist parameters $\lambda \in [0,n], r\in \N$ such that $P_{n,\lambda,r}$ is $(\eps,\delta)$ private and for any $X\in(0,1)^n$
\[
\pr{}{\left| P^\R_{n,\lambda,r}(X) - f(X) \right| > \frac{122}{\eps} \log\frac{8}{\delta} \sqrt{\log\frac{2}{\beta}}} < 2\beta
\]
\end{coro}

\begin{proof}
Let $r = \lceil \eps \sqrt{n}\rceil$. Define
\begin{mymath}
\eps_0 = \frac{\eps}{\sqrt{8 r \log(2/\delta)}} \ifnum\oakland=0~~~~~~~~~~~~\else\textrm{, and }\fi \delta_0 = \frac{\delta}{2r}
\end{mymath}
In an identical fashion to Lemma \ref{lem:generate-lambda}, assign $\lambda$ such that $P^{0/1}_{n,\lambda}$ satisfies $(\eps_0,\delta_0)$ privacy. From Corollary \ref{coro:real-sum-private}, this in turn implies $P_{n,\lambda,r}$ is $(\eps,\delta)$ differentially private.

For these values of $r,\lambda$, Theorem \ref{thm:real-sum-error} bounds the error as
\ifnum\oakland=1
\begin{equation}
\label{eq:real-sum-accurate}
\frac{\sqrt{2}}{r}\sqrt{n\log \frac{2}{\beta}} + \cdot \frac{n}{n-\lambda} \cdot \sqrt{2 \frac{\lambda}{r} \log \frac{2}{\beta}}
\end{equation}
with probability at least $1-2\beta$
\else
\begin{equation}
\label{eq:real-sum-accurate}
\pr{}{ \left|P^\R_{n,\lambda,r}(X) - f(X) \right| \geq \frac{\sqrt{2}}{r}\sqrt{n\log \frac{2}{\beta}} + \cdot \frac{n}{n-\lambda} \cdot \sqrt{2 \frac{\lambda}{r} \log \frac{2}{\beta}}} \leq 2\beta
\end{equation}
\fi

It is immediate from substitution of $r$ that
\begin{equation}
\label{eq:real-sum-error-eps-1}
\frac{\sqrt{2}}{r}\sqrt{n\log\frac{2}{\beta}} = \frac{\sqrt{2}}{\eps}\sqrt{\log\frac{2}{\beta}}
\end{equation}

Following the same steps as the proof of Corollary \ref{coro:bit-sum-accurate-concrete}, it can be shown that
\ifnum\oakland=1
\begin{align*}
\frac{n}{n-\lambda} \cdot \sqrt{2 \frac{\lambda}{r} \log \frac{2}{\beta}} &\leq 
 \frac{120}{\eps} \log\frac{8}{\delta} \sqrt{\log\frac{2}{\beta}}
\stepcounter{equation}\tag{\theequation}\label{eq:real-sum-error-eps-2}
\end{align*}
\else
\begin{align*}
\frac{n}{n-\lambda} \cdot \sqrt{2 \frac{\lambda}{r} \log \frac{2}{\beta}} &\leq \frac{30}{\eps_0}\sqrt{\frac{1}{r} \log \left(\frac{4}{\delta_0}\right)\log\frac{2}{\beta}}\\
&= 30\cdot \frac{\sqrt{8r\log\frac{2}{\delta}}}{\eps} \sqrt{\frac{1}{r} \log \left(\frac{4}{\delta_0}\right)\log\frac{2}{\beta}} \tag{Defn. of $\eps_0$}\\
&= \frac{30}{\eps} \sqrt{8\log\left(\frac{2}{\delta}\right) \log \left(\frac{4}{\delta_0}\right)\log\frac{2}{\beta}}\\
&= \frac{30}{\eps} \sqrt{8\log\left(\frac{2}{\delta}\right) \log \left(\frac{8r}{\delta_0}\right)\log\frac{2}{\beta}} \tag{Defn. of $\delta_0$} \\
&= \frac{30}{\eps} \sqrt{8\log\left(\frac{2}{\delta}\right) \log \left( \frac{8\eps\sqrt{n}}{\delta}\right) \log\frac{2}{\beta}} \tag{Defn. of $r$}\\
&< \frac{30}{\eps} \sqrt{16 \log\left(\frac{2}{\delta}\right) \log\left(\frac{8}{\delta}\right) \log\frac{2}{\beta}} \tag{$\eps < 1,\delta < 8 / \sqrt{n}$}\\
&= \frac{120}{\eps} \log\frac{8}{\delta} \sqrt{\log\frac{2}{\beta}}
\stepcounter{equation}\tag{\theequation}\label{eq:real-sum-error-eps-2}
\end{align*}
\fi

Observe $\frac{\mathrm{Eq}\eqref{eq:real-sum-error-eps-2}}{\mathrm{Eq}\eqref{eq:real-sum-error-eps-1}} = 60\sqrt{2}\log\frac{8}{\delta} < 85\log\frac{8}{\delta}$. Hence,
\ifnum\oakland=1
\begin{align*}
\mathrm{Eq}\eqref{eq:real-sum-error-eps-1} + \mathrm{Eq}\eqref{eq:real-sum-error-eps-2} &<  \frac{122}{\eps} \log\frac{8}{\delta} \sqrt{\log\frac{2}{\beta}}
\end{align*}
\else
\begin{align*}
\mathrm{Eq}\eqref{eq:real-sum-error-eps-1} + \mathrm{Eq}\eqref{eq:real-sum-error-eps-2} &< \mathrm{Eq}\eqref{eq:real-sum-error-eps-1}\cdot ( 1 + 85\log\frac{8}{\delta})\\
&< \mathrm{Eq}\eqref{eq:real-sum-error-eps-1} \cdot 86\log\frac{8}{\delta} \tag{$\delta < 1$}\\
&= \frac{86\sqrt{2}}{\eps} \log\frac{8}{\delta} \sqrt{\log\frac{2}{\beta}}\\
&< \frac{122}{\eps} \log\frac{8}{\delta} \sqrt{\log\frac{2}{\beta}}
\end{align*}
\fi
The above can be substituted into \eqref{eq:real-sum-accurate} and we arrive at
\[
\pr{}{ \left|P_{n,\lambda,r}(X) - f(X) \right| \geq \frac{122}{\eps} \log\frac{8}{\delta} \sqrt{\log\frac{2}{\beta}}} \leq 2\beta
\]
which is precisely the target claim. 
\end{proof}

\section{From Approximate DP to Pure DP for Local Protocols} \label{sec:approx-to-pure}
In this section, we show that for any $(\eps,\delta)$ private local protocol, there exists an $(8\eps,0)$ private counterpart with roughly the same accuracy guarantees.  \cite{bns18} proved this theorem when $\eps \leq 1/4$, but we need the theorem when $\eps \gg 1$.  Our proof follows their approach almost exactly, but we include it for completeness to verify that their result can be modified to hold for larger $\eps$.

\begin{thm} [Extension of~\cite{bns18}]
\label{thm:transform-appendix}
Let $P_n = (R_n,A_n)$ be a local protocol for $n \geq 3$ users that is $(\eps,\delta)$ differentially private and $(\alpha,\beta)$ accurate with respect to $f$. If $\eps > 2/3$ and
\[
\delta < \frac{\beta}{8n\ln(n/\beta) e^{6\eps}}
\]
then there exists another local protocol that is $(8\eps,0)$ differentially private and $(\alpha,4\beta)$ accurate with respect to $f$.
\end{thm}

The conditions $\eps > 2/3$ and $n > 3$ are not essential, but are used to simplify the statement.
We will prove Theorem \ref{thm:transform-appendix} by construction: given local randomizer $R$, Algorithm \ref{alg:transform}  transforms it into another randomizer $R_{k,T}$. The parameters will be set later to achieve the desired privacy and accuracy.


\begin{algorithm}
\caption{A local randomizer $R_{k,T}$}
\label{alg:transform}
\DontPrintSemicolon
\SetKwInOut{Input}{Input}
\SetKwInOut{Output}{Output}

\KwIn{$x \in \cX$; parameters $k \in (0, 2e^{-2\eps})$ and $T \in \N$; black-box access to $R:\cX \rightarrow \cY$}
\KwOut{$\mathbf{y}_{k,T} \in \cY$}

\BlankLine
Let $c$ be some (publicly known) fixed element of $\cX$.\\
Define $\mathit{GoodInt} := \left[\half\exp(-2\eps), \half\exp(2\eps)\right]$\\
\For{$t\in[T]$}{
  $\mathbf{v}_t \gets R(c)$\\
  $\mathbf{p}_t \gets \half \frac{\pr{}{R(x)=\mathbf{v}_t}}{\pr{}{R(c)=\mathbf{v}_t}}$\\
  \lIf{$\mathbf{p}_t \notin \mathit{GoodInt}$}{
	$\mathbf{p}_t \gets \half$
  }{ $\mathbf{b}_t \gets \mathit{Ber}(\mathbf{p}_t \cdot k)$}
}
\lIf{$\exists t~ \mathbf{b}_{t}=1$}{ Sample $\mathbf{j}$ uniformly over $\{t\in[T]~|~b_{t}=1\}$ }
\lElse{Sample $\mathbf{j}$ uniformly over $[T]$}

Return $\mathbf{y}_{k,T} \gets \mathbf{v_j}$

\end{algorithm}

\subsection{Privacy Analysis}
First, we establish that the transformed local randomizer is indeed $(8\eps,0)$-differentially private.

\begin{clm}
\label{clm:transform-is-private}
For any $\eps>0$, any $(\eps,\delta)$-differentially private algorithm $R$, any $k \in (0,2e^{-2\eps})$, and any $T\in \N$, the algorithm $R_{k,T}$ is $(8\eps, 0)$ differentially private.
\end{clm}

\begin{proof}
Define $L := \half\exp(-2\eps) \cdot k$ and $U := \half\exp(2\eps) \cdot k$. Note that $[L,U] \subset [0,1]$ and $\pr{}{\mathbf{b_{i,j}} = 1} \in [L,U]$.


Fix $x\sim x'\in\cX$, $j\in [T]$, and $y \in \cY$. Let $\cV_{y} = \{V\in \cY^T \mid \exists j ~ v_j = y\}$.
\begin{align*}
\pr{}{R_{k,T}(x) = y} &= \pr{}{R_{k,T}(x) = y \cap \mathbf{V} \in \cV_{y} }\\
&= \sum_{V\in \cV_{y}} \pr{}{\mathbf{V} = V} \cdot \pr{}{R_{k,T}(x)=y \mid \mathbf{V} = V}\\
&= \sum_{V\in \cV_{y}} \pr{}{\mathbf{V} = V} \cdot \left( \sum_{ \{j \mid v_j = y \}} \pr{}{\mathbf{j}=j \mid \mathbf{V} = V} \right) \stepcounter{equation} \label{eq:rkt-event} \tag{\theequation}
\end{align*}


Recall that $\mathbf{b}_j$ takes a random binary value and the distribution of $\mathbf{j}$ is dependent on these bits. In the analysis below, we omit the condition that $\mathbf{V}=V$ for brevity.
\begin{equation}
\label{eq:j}
\pr{}{\mathbf{j}=j} = \pr{}{\mathbf{j}=j \cap \mathbf{b_{j}} = 0} + \pr{}{\mathbf{j}=j \cap \mathbf{b_{j}} = 1}
\end{equation}
We'll upper bound each summand separately. If $b_{j}$ is set to zero, then the only way to for the user to choose $j$ is for all other bits $b_{t}$ to be zero as well. And in that case, the choice is uniform over $[T]$:

\begin{align*}
\pr{}{\mathbf{j}=j \cap \mathbf{b_{j}} = 0} &= \frac{1}{T}\cdot \pr{}{\forall t~ \mathbf{b_{t}} = 0}\\
&= \frac{1}{T}\cdot \prod_{t=1}^{T}\pr{}{\mathbf{b_{t}} = 0} \tag{Independence}\\
&\leq \frac{1}{T}\cdot \prod_{t=1}^{T} (1-L) \tag{$\pr{}{\mathbf{b_{t}} = 1} \geq L$}\\
&= \frac{1}{T}\cdot (1-L)^T \stepcounter{equation} \tag{\theequation} \label{eq:j-1}
\end{align*}

If $b_{j}$ is set to one, $\mathbf{j}$ is uniform over the bits set to one, itself a random variable:
\begin{align*}
\pr{}{\mathbf{j}=j \cap \mathbf{b_{j}} = 1} &= \pr{}{\mathbf{b_{j}}=1}\cdot \pr{}{\mathbf{j}=j \mid \mathbf{b_{j}} = 1}\\
&= \pr{}{\mathbf{b_{j}}=1}\cdot \sum_{s=1}^{T} \frac{1}{s}\cdot \pr{}{\sum_{t\neq j} \mathbf{b_{t}} =s-1}\\
&= \pr{}{\mathbf{b_{j}}=1}\cdot \sum_{s=0}^{T-1} \frac{1}{s+1}\cdot \pr{}{\sum_{t\neq j} \mathbf{b_{t}} =s}\\
&= \pr{}{\mathbf{b_{j}}=1}\cdot \ex{}{\frac{1}{1+ \sum_{t\neq j} \mathbf{b_{t}} }} \stepcounter{equation} \tag{\theequation} \label{eq:j-2-helper}
\end{align*}

Observe that the term $\sum_{t\neq j} \mathbf{b_{t}}$ is a sum of Bernoulli random variables, with different expectations but all residing in $[L,U]$. As a corollary of Claim 6.3 in \cite{bns18}, we have the following:
\begin{clm}
If random variables $\mathbf{b_1},\dots, \mathbf{b_T}$ are each drawn independently from $\mathit{Ber}(p_1),\dots,\mathit{Ber}(p_T)$ where $L \leq p_t \leq U$ for every $t \in [T]$, then $$\ex{}{\frac{1}{1+\mathit{Bin}(T,U)}} \leq \ex{}{\frac{1}{1+\sum \mathbf{b_t}}} \leq \ex{}{\frac{1}{1+\mathit{Bin}(T,L)}}.$$
\end{clm}

Hence,
\begin{align*}
\eqref{eq:j-2-helper} &\leq \pr{}{\mathbf{b_{j}}=1}\cdot \ex{}{\frac{1}{1+ \mathit{Bin}(T-1,L) }}\\
&= \pr{}{\mathbf{b_{j}}=1}\cdot \sum_{s=0}^{T-1}\cdot \frac{1}{s+1} \binom{T-1}{s}\cdot L^s\cdot (1-L)^{T-s-1}\\
&= \pr{}{\mathbf{b_{j}}=1}\cdot \frac{1}{TL} \left(1-(1-L)^T\right)\\
&\leq \frac{U}{TL} \left(1-(1-L)^T\right) \stepcounter{equation} \tag{\theequation} \label{eq:j-2}
\end{align*}

From \eqref{eq:j-1}, \eqref{eq:j-2}, and \eqref{eq:j},
\begin{align*}
\label{eq:j-upper}
\pr{}{\mathbf{j}=j} &\leq  \frac{U}{TL} \left(1-(1-L)^T\right) + \frac{1}{T} (1-L)^T \\
&\leq  \frac{U}{TL} \left(1-(1-L)^T\right) + \frac{U}{TL} (1-L)^T \tag{$U > L$}\\
&= \frac{U}{TL} \stepcounter{equation} \tag{\theequation} \label{eq:j-3}
\end{align*}

Recall \eqref{eq:rkt-event}:
\begin{align*}
\pr{}{R_{k,T}(x) = y} &= \sum_{V\in \cV_{y}} \pr{}{\mathbf{V} = V} \cdot \left( \sum_{ \{j \mid v_j = y \}} \pr{}{\mathbf{j}=j \mid \mathbf{V} = V} \right) \\
&< \frac{U}{TL} \cdot \sum_{V\in \cV_y} \pr{}{\mathbf{V} = V} \cdot \{\# v_j = y\} \tag{From \eqref{eq:j-3}}\\
&= \frac{U}{TL} \cdot \ex{\mathbf{V}\gets R(c)^T}{\# \mathbf{v}_j = y} \stepcounter{equation} \tag{\theequation} \label{eq:rkt-upper}
\end{align*}
where we use ${\# \mathbf{v}_j = y}$ to indicate the number of elements of $\mathbf{V}$ that have value $y$.

We remark that the distribution of $\mathbf{V}$ is wholly independent of private value $x$. Hence, by a completely symmetric series of steps,
\begin{equation}
\label{eq:rkt-lower}
\pr{}{R_{k,T}(x') = y} > \frac{L}{TU} \cdot \ex{\mathbf{V}\gets R(c)^T}{\# \mathbf{v}_j = y}
\end{equation}

From \eqref{eq:rkt-upper} and \eqref{eq:rkt-lower},
\begin{align*}
\frac{\pr{}{R_{k,T}(x) = y}}{\pr{}{R_{k,T}(x') = y}} &< \frac{U^2}{L^2}\\
&= \frac{(1/4)\cdot \exp(4\eps)\cdot k^2}{(1/4)\cdot \exp(-4\eps)\cdot k^2}\\
&= \exp(8\eps)
\end{align*}

which completes the proof.
\end{proof}

\subsection{Accuracy Analysis}
Next we show that, for suitable parameters, the protocol $P_{n,k,T} = (R_{n,k,T}, A_{n})$ remains essentially as accurate as $P_{n} = (R_{n},A_{n})$.
\begin{clm}
\label{clm:transform-is-accurate}
If $P_{n} = (R_{n},A_{n})$ is $(\alpha,\beta)$-accurate for $f$, and $R_{n}$ is $(\eps,\delta)$-differentially private for $\eps > 2/3$ and
\[
\delta < \frac{\beta}{8n\ln(n/\beta)}\cdot \frac{1}{\exp(6\eps)}
\]
then there exists $T\in \N, k \in (0,2e^{-2\eps})$ such that the local protocol $P_{n,k,T} = (R_{n,k,T},A_{n})$ is $(\alpha,4\beta)$ accurate for $f$.
\end{clm}
Claim \ref{clm:transform-is-private} and Claim \ref{clm:transform-is-accurate} together imply Theorem \ref{thm:transform}. 

\medskip

For the purposes of this section, fix any $X\in \cX^n$ where $X=(x_1,\dots, x_n)$. Let $\mathbf{y_{k,T}[i]}$ denote\footnote{Brackets are used for indices in order to avoid collision with parameters $k,T$ in the subscript} the random output of $R_{n,k,T}(x_i)$ and let $\mathbf{Y_{k,T}}$ denote the ordered set $\mathbf{y_{k,T}[1]}, \dots, \mathbf{y_{k,T}[n]}$. Let $\mathbf{y}[i]$ denote the random output of $R_n(x_i)$ and let $\mathbf{Y}$ denote the ordered set $\mathbf{y}[1],\dots,\mathbf{y}[n]$.

As described in Algorithm \ref{alg:transform}, $R_{k,T}$ defines variables $\mathbf{b}_t,\mathbf{v}_t$ for every $t\in [T]$. In the context of $P_n$, there are $2n\cdot T$ such variables, $2T$ for each user $i$: we use $\mathbf{b}_{i,t},\mathbf{v}_{i,t}$ to disambiguate between users.

As a first step to proving Claim \ref{clm:transform-is-accurate}, we show that the distribution of $\mathbf{Y_{k,T}}$ is similar to that of $\mathbf{Y}$ (Claim \ref{clm:transform-is-accurate-1}). Then we show that running the same analysis function $A_n$ on both $\mathbf{Y}$ and $\mathbf{Y_{k,T}}$ yields similar accuracy guarantees (Claim \ref{clm:transform-is-accurate-2}). The notion of ``similar'' is statistical distance in terms of parameters $n,\eps,\delta,k,T$: if the parameters are constrained, then the distance simplifies to $3\beta$. Because $\beta + 3\beta = 4\beta$, we have $(\alpha,4\beta)$ accuracy (Claim \ref{clm:bounding-the-gap}). We close the section by a particular setting of $k,T$ to achieve such a bound.

As we have stated, we start by relating $\mathbf{Y_{k,T}}$ to $\mathbf{Y}$:

\begin{clm}
\label{clm:transform-is-accurate-1}
If $\eps > 0, 0 < \delta < \frac{1-\exp(-\eps)}{4\exp(\eps)n}, T\in \N, k \in (0,2e^{-2\eps})$, then for any $\cW\subseteq \univy^n$,
\begin{equation}
\label{eq:transform-distance}
\pr{}{\mathbf{Y_{k,T}} \in \cW} < \pr{}{\mathbf{Y} \in \cW} + n\cdot \left(1-\half\exp(-2\eps)\cdot k \right)^T + \frac{2n\delta\exp(\eps)}{1-\exp(-\eps)} (T + 2)
\end{equation}
\end{clm}

\begin{proof}
Consider an execution of the protocol $P_{n,k,T}(X)$. Let $E_1$ denote the event that for some user $i$, all bits $\mathbf{b}_{i,j}$ are set to 0:
\begin{align*}
\pr{}{E_1} &= \pr{}{\exists i~ \forall t~ \mathbf{b}_{i,t}=0}\\
&\leq n\cdot \max_{i\in [n]} \pr{}{\forall t~ \mathbf{b}_{i,t}=0} \tag{Union bound}\\
&\leq n\cdot \left(1-\half\exp(-2\eps) \cdot k \right)^T \stepcounter{equation} \tag{\theequation} \label{eq:transform-error-1}
\end{align*}

Recall that $\mathit{GoodInt}=[\half e^{-2\eps},\half e^{2\eps}]$. For any $x,x' \in \univx$, let $\mathit{Good}(x,x')\subset \univy$ denote the set consisting of all $\tau$ satisfying $\half \frac{\prob(R_n(x')=\tau)}{\prob(R_n(x)=\tau)} \in \mathit{GoodInt}$. The following is a property of private algorithms:
\begin{lem}
\label{lem:leakage-likelihood}
Fix a value of $c \in \univx$. If $R:\univx \rightarrow \univy$ is $(\eps,\delta)$ differentially private, then for any $x\in \univx$,
\[
\pr{}{R(c)\notin \mathit{Good}(c,x)} \leq \frac{2\delta\exp(\eps)}{1-\exp(-\eps)}
\]
\end{lem}

A proof can be found in \cite[Claim 5.4]{BassilyS15}.  Let $E_2$ denote the event that for some $i$ and some $t$, $\mathbf{v_{i,t}} \notin \mathit{Good}(c,x_i)$.
\begin{align*}
\pr{}{E_2} &= \sum_{i=1}^n \sum_{t=1}^T \pr{}{\mathbf{v}_{i,t} \notin \mathit{Good}(c,x_i) } \tag{Independence}\\
&\leq nT\cdot \max_{i\in [n]}\left( \pr{}{R(c) \notin \mathit{Good}(c,x_i)} \right) \tag{$\mathbf{v}_{i,t} \gets R(c)$}\\
&\leq \frac{2nT\delta\exp(\eps)}{1-\exp(-\eps)} \stepcounter{equation} \label{eq:transform-error-2} \tag{\theequation}
\end{align*}
The second inequality is an application of Lemma \ref{lem:leakage-likelihood} to $(\eps,\delta)$ private $R_n$.


Fix any $\cW\subset \univy^n$.
\begin{align*}
\pr{}{\mathbf{Y_{k,T}} \in \cW} &= \pr{}{\mathbf{Y_{k,T}} \in \cW \cap E_1} + \pr{}{\mathbf{Y_{k,T}} \in \cW \cap E_2}\\
&~~~~~+ \pr{}{\mathbf{Y_{k,T}} \in \cW \cap (\neg E_1 \cap \neg E_2)}\\
&\leq \pr{}{E_1} + \pr{}{E_2} + \pr{}{\mathbf{Y_{k,T}} \in \cW \cap (\neg E_1 \cap \neg E_2)}\\
&\leq n\cdot \left(1-\half\exp(-2\eps) \cdot k \right)^T + \frac{2nT\delta\exp(\eps)}{1-\exp(-\eps)}\\
&~~~~~ + \pr{}{\mathbf{Y_{k,T}} \in \cW \cap (\neg E_1\cap \neg E_2)} \stepcounter{equation} \tag{\theequation} \label{eq:transform-error}
\end{align*}
The second inequality is simply substitution of \eqref{eq:transform-error-1} and \eqref{eq:transform-error-2}

Fix some $i\in[n]$. Notice that if $\neg E_2$ occurs, then it must be the case that for every $t \in [T]$, $\mathbf{v}_{i,t} \in \mathit{Good}(c,x_i)$. Because $\mathbf{y_{k,T}}[i]$ is selected from $\mathbf{v_{i,t}}$, it must be the case that $\mathbf{y_{k,T}}[i]$ has to lie in $\mathit{Good}(c,x_i)$.

Let $\mathit{Good}\subset \univy^n$ denote $\mathit{Good}(c,x_1)\times \dots \times \mathit{Good}(c,x_n)$. If $\neg E_2$ occurs, then $\mathbf{Y_{k,T}} \in \mathit{Good}$. We use this to analyze the third summand in \eqref{eq:transform-error}:
\begin{align*}
\pr{}{\mathbf{Y_{k,T}} \in \cW \cap (\neg E_1\cap \neg E_2)} &= \sum_{W\in \cW}\pr{}{\mathbf{Y_{k,T}} = W \cap (\neg E_1\cap \neg E_2)}\\
&= \sum_{W\in \cW \cap \mathit{Good}}\pr{}{\mathbf{Y_{k,T}} = W \cap (\neg E_1\cap \neg E_2)}\\
&= \sum_{W\in \cW \cap \mathit{Good}}\prod_{i=1}^n \pr{}{\mathbf{y_{k,T}[i]} = w[i] \cap (\neg E_1\cap \neg E_2)} \tag{Independence}\\
&\leq \sum_{W\in \cW \cap \mathit{Good}}\prod_{i=1}^n \pr{}{\mathbf{y_{k,T}[i]} = w[i] \mid (\neg E_1\cap \neg E_2)}  \stepcounter{equation} \tag{\theequation} \label{eq:transform-error-3a}
\end{align*}

We will later prove the following equivalence:
\begin{clm}
\label{clm:transformIsAccurate-helper2} For any $\eps> 0, k\in (0,2\exp(-2\eps)), T\in \N$ and $c,x\in \cX$,
\[
\pr{\mathbf{y_{k,T}} \gets R_{n,k,T}(x)}{\mathbf{y_{k,T}}=g \mid (\neg E_1\cap \neg E_2)} = \pr{\mathbf{y} \gets R_n(x)}{\mathbf{y} = g \mid \mathbf{y} \in \mathit{Good}(c,x)}
\]
for any $g \in \mathit{Good}(c,x)$
\end{clm}

By substitution,
\begin{align*}
\eqref{eq:transform-error-3a} &= \sum_{W\in \cW \cap \mathit{Good}}\prod_{i=1}^n \pr{}{\mathbf{y}[i]=w[i] \mid y[i] \in\mathit{Good}(c,x_i)}\\
&= \sum_{W\in \cW \cap \mathit{Good}}\pr{}{\mathbf{Y}=W \mid \mathbf{Y} \in\mathit{Good}} \tag{Independence}\\
&= \pr{}{\mathbf{Y} \in \cW \mid \mathbf{Y} \in\mathit{Good}}\\
&\leq \frac{1}{\pr{}{\mathbf{Y} \in\mathit{Good}}}\cdot \pr{}{\mathbf{Y}\in \cW }\\
&= \frac{1}{1-\pr{}{\mathbf{Y} \notin\mathit{Good}}}\cdot \pr{}{\mathbf{Y}\in \cW } \stepcounter{equation}\tag{\theequation}\label{eq:transform-error-3b}
\end{align*}

Notice that $\mathbf{Y} \notin\mathit{Good}$ when, for some $i$, $\half \frac{\prob(R_n(x_i)=\mathbf{y}[i])}{\prob(R_n(c)=\mathbf{y}[i])} \notin \mathit{GoodInt}$. We obtain $\pr{}{\mathbf{Y} \notin\mathit{Good}} \leq \frac{2n\delta\exp(\eps)}{1-\exp(-\eps)}$ by an argument similar\footnote{Notice the absence of $T$: to arrive at \eqref{eq:transform-error-2}, we union bound over $|\mathbf{V}| =nT$ random variables but here $|\mathbf{Y}|=n$} to that of \eqref{eq:transform-error-2}. Therefore,
\begin{align*}
\eqref{eq:transform-error-3b} &\leq \frac{1}{1-\frac{2n\delta\exp(\eps)}{1-\exp(-\eps)}}\cdot \pr{}{\mathbf{Y}\in \cW } \\
&\leq \left(1+\frac{4n\delta\exp(\eps)}{1-\exp(-\eps)}\right)\cdot \pr{}{\mathbf{Y} \in \cW } \tag{$\frac{2n\delta\exp(\eps)}{1-\exp(-\eps)} < \half$}\\
&\leq \pr{}{\mathbf{Y} \in \cW } + \frac{4n\delta\exp(\eps)}{1-\exp(-\eps)}
\end{align*}

When we return to \eqref{eq:transform-error}, we have
\begin{align*}
\pr{}{\mathbf{Y_{k,T}} \in \cW} \leq& \pr{}{\mathbf{Y} \in \cW } + n\cdot \left(1-\half\exp(-2\eps) \cdot k \right)^T\\
&~~~~~+ \frac{2nT\delta\exp(\eps)}{1-\exp(-\eps)} + \frac{4n\delta\exp(\eps)}{1-\exp(-\eps)}
\end{align*}
which is equivalent to \eqref{eq:transform-distance}. This concludes the proof, modulo Claim \ref{clm:transformIsAccurate-helper2}.
\end{proof}

Here, we prove Claim \ref{clm:transformIsAccurate-helper2}.
\begin{proof}[Proof of Claim \ref{clm:transformIsAccurate-helper2}]
Fix any $\eps> 0, k\in (0,2\exp(-2\eps)), T\in \N, (c,x)\in \cX^2$ and $g \in \mathit{Good}(c,x)$. Sample $\mathbf{y}$ from $R_n(x)$ and $\mathbf{y_{k,T}}$ from $R_{n,k,T}(x)$.

By a corresponding argument advanced by \cite{bns18},
\begin{equation}
\label{eq:transformIsAccurate-helper2}
\pr{}{\mathbf{y_{k,T}}=g \mid (\neg E_1\cap \neg E_2)} = \frac{\pr{}{\mathbf{b}_1=1 \cap \mathbf{v}_1 = g}}{\pr{}{\mathbf{b}_1 =1 \cap \mathbf{v}_1 \in \mathit{Good}(c,x)}}
\end{equation}

We first expand the numerator:
\begin{align*}
\pr{}{\mathbf{b}_1 =1 \cap \mathbf{v}_1 = g} 
&= \pr{}{\mathbf{v}_1 = g} \cdot \pr{}{\mathbf{b_{1}}=1 \mid \mathbf{v}_1 = g}\\
&= \pr{}{\mathbf{v}_1 = g} \cdot \half \frac{\pr{}{\mathbf{y}=g}}{\pr{}{\mathbf{v_{1}}=g}}\cdot k\tag{Defn. of $R_{k,T}$}\\
&= \half \pr{}{\mathbf{y}=g}\cdot k \stepcounter{equation} \tag{\theequation} \label{eq:transformIsAccurate-helper2-num}
\end{align*}

We now analyze the denominator:
\begin{align*}
\pr{}{\mathbf{b}_1 =1 \cap \mathbf{v}_1 \in \mathit{Good}(c,x)} &= \sum_{\tau \in \mathit{Good}(c,x)} \pr{}{\mathbf{v}_1 =\tau}\cdot \pr{}{\mathbf{b}_1 =1 \mid \mathbf{v}_1 = \tau}\\
&= \sum_{\tau \in \mathit{Good}(c,x)} \pr{}{\mathbf{v}_1 =\tau}\cdot \half \frac{\pr{}{\mathbf{y}=\tau}}{\pr{}{\mathbf{v_{1}}=\tau}} \cdot k \tag{Defn. of $R_{k,T}$}\\
&= \half \cdot \sum_{\tau \in \mathit{Good}(c,x)} \pr{}{\mathbf{y}=\tau} \cdot k\\
&= \half \pr{}{\mathbf{y}\in \mathit{Good}(c,x)}\cdot k \stepcounter{equation} \tag{\theequation} \label{eq:transformIsAccurate-helper2-den}
\end{align*}

Therefore,
\begin{align*}
Eq\eqref{eq:transformIsAccurate-helper2} &= \frac{Eq\eqref{eq:transformIsAccurate-helper2-num}}{Eq\eqref{eq:transformIsAccurate-helper2-den}}= \frac{\half \pr{}{\mathbf{y}=g}\cdot k}{\half \pr{}{\mathbf{y}\in \mathit{Good}(c,x)}\cdot k}\\
&= \pr{}{\mathbf{y}=g \mid \mathbf{y} \in\mathit{Good}(c,x)}
\end{align*}
which completes the proof.
\end{proof}



The preceding bound on the statistical distance between $R_n(X),R_{n,k,T}(X)$ implies a bound on the error of the transformed protocol $P_{n,k,T}$:

\begin{clm}
\label{clm:transform-is-accurate-2}
Suppose $P_n=(R_n,A_n)$ is $(\eps,\delta)$ differentially private and $(\alpha,\beta)$ accurate. If $\eps > 0, \delta \in (0, \frac{1-\exp(-\eps)}{4\exp(\eps)n}), T\in \N, k \in (0,2e^{-2\eps})$, then $P_{n,k,T}=(R_{n,k,T},A_n)$ is $(\alpha,\beta_{k,T})$ accurate where
\begin{equation}
\label{eq:bkt}
\beta_{k,T} = \beta +  n \cdot \left(1-\half\exp(-2\eps)\cdot k \right)^T + \frac{2n\delta\exp(\eps)}{1-\exp(-\eps)}(T+2)
\end{equation}
\end{clm}

\begin{proof}
If $A_n$ is deterministic, the claim is immediate from Claim \ref{clm:transform-is-accurate-1}. Otherwise, the randomness of $\mathbf{u}$ is sourced from both $A_n$ and $R_n$.

For any $X=(x_1,\dots,x_n)\in \cX^n$, we again use $\mathbf{Y}$ to denote the random variable output by $R_n(x_1),\dots, R_n(x_n)$, likewise $\mathbf{Y_{k,T}}$ for the random variable output by $R_{n,k,T}(x_1),\dots, R_{n,k,T}(x_n)$. We will use $\mathbf{u}$ to denote the random variable $A_n(\mathbf{Y})$, which is the output of the original protocol, and $\mathbf{u_{k,T}}$ to denote the random variable $A_n(\mathbf{Y_{k,T}})$, which is the output of the transformed protocol.

For any $Y \in \univy^n$, let $\Delta_Y := \pr{}{\mathbf{Y_{k,T}} = Y}-\pr{}{\mathbf{Y} = Y}$. Let $I$ denote the subset of $\univy^n$ containing exactly those $Y$ such that $\pr{}{\mathbf{Y_{k,T}} = Y} > \pr{}{\mathbf{Y} = Y}$; equivalently, those $Y$ where $\Delta_Y > 0$. We will use $I$ to analyze the probability of exceeding $\alpha$ error: assuming we are interested in approximating a real-valued $f(X)$,
\begin{align*}
\pr{}{|\mathbf{u_{k,T}} - f(X)| > \alpha} &= \pr{}{|\mathbf{u_{k,T}}-f(X)| > \alpha \cap \mathbf{Y_{k,T}}\in I}\\
&+ \pr{}{|\mathbf{u_{k,T}}-f(X)| > \alpha \cap \mathbf{Y_{k,T}}\notin I} \label{eq:mind-the-gap} \stepcounter{equation} \tag{\theequation}
\end{align*}

We bound each term in the sum separately.
\begin{align*}
\pr{}{|\mathbf{u_{k,T}}-f(X)| > \alpha \cap \mathbf{Y_{k,T}}\in I}  &= \sum_{Y \in I} \pr{}{|A_n(Y)-f(X)| > \alpha} \cdot \pr{}{\mathbf{Y_{k,T}} = Y} \\
&= \sum_{Y \in I} \pr{}{|A_n(Y)-f(X)| > \alpha} \cdot \left( \Delta_Y + \pr{}{\mathbf{Y} = Y}\right)\\
&< \sum_{Y \in I} \pr{}{|A_n(Y)-f(X)| > \alpha} \cdot \pr{}{\mathbf{Y} = Y} + \Delta_Y\\
&= \sum_{Y \in I} \pr{}{|A_n(Y)-f(X)| > \alpha \cap \mathbf{Y}=Y} + \Delta_Y\\
&= \pr{}{|\mathbf{u}- f(X)| > \alpha \cap \mathbf{Y} \in I} + \left(\sum_{Y \in I} \Delta_Y \right) \stepcounter{equation} \label{eq:mind-the-gap-1} \tag{\theequation}
\end{align*}
where the inequality comes from the fact that $\Delta_Y > 0$.

\begin{align*}
\pr{}{|\mathbf{u_{k,T}}-f(X)| > \alpha \cap \mathbf{Y_{k,T}}\notin I} &= \sum_{Y \notin I} \pr{}{|A_n(Y)-f(X)| > \alpha} \cdot \pr{}{\mathbf{Y_{k,T}} = Y}\\
&\leq \sum_{Y \notin I} \pr{}{|A_n(Y)-f(X)| > \alpha} \cdot \pr{}{\mathbf{Y} = Y}\\
&= \pr{}{|\mathbf{u}- f(X)| > \alpha \cap \mathbf{Y}\notin I} \stepcounter{equation} \label{eq:mind-the-gap-2} \tag{\theequation}
\end{align*}
The inequality comes from the definition of $\neg I$.

From \eqref{eq:mind-the-gap}, \eqref{eq:mind-the-gap-1}, and \eqref{eq:mind-the-gap-2} we have
\begin{align*}
\pr{}{|\mathbf{u_{k,T}} - f(X)| > \alpha} &< \pr{}{|\mathbf{u}- f(X)| > \alpha} + \left(\sum_{Y \in I} \Delta_Y \right)\\
&\leq \beta + \left(\sum_{Y \in I} \Delta_Y \right) \tag{$P_n$ is $(\alpha,\beta)$-accurate}\\
&= \beta + \left( \pr{}{\mathbf{Y_{k,T}} \in I} - \pr{}{\mathbf{Y} \in I} \right) \\
&< \beta + n \cdot \left(1-\half\exp(-2\eps)\cdot k \right)^T + \frac{2n\delta\exp(\eps)}{1-\exp(-\eps)}(T+2) \tag{Claim \ref{clm:transform-is-accurate-1}}
\end{align*}

This completes the proof.
\end{proof}

Finally, we show that the preceding error probability simplifies to $4\beta$ provided that parameters $n,\eps,\delta$ obey some constraints.
\begin{clm}
\label{clm:bounding-the-gap}
For all $\eps > 2/3$, $k \in (2\exp(-3\eps) , 2\exp(-2\eps))$, $n \geq 3$ and 
\begin{equation}
0 < \delta < \frac{\beta}{8n\ln(n/\beta)}\cdot \frac{1}{\exp(6\eps)} \label{eq:upper-bound-delta}
\end{equation}
then there exists $T\in \N$ such that
\begin{equation}
\label{eq:three-beta-terms}
n \cdot \left(1-\half\exp(-2\eps)\cdot k \right)^T + \frac{2nT\delta\exp(\eps)}{1-\exp(-\eps)} + \frac{4n\delta\exp(\eps)}{1-\exp(-\eps)} < 3\beta
\end{equation}

\end{clm}

\begin{proof}

\eqref{eq:three-beta-terms} holds when each term in the sum is $\leq \beta$. 

We begin with the term $n \cdot \left(1-\half\exp(-2\eps)\cdot k \right)^T$. Because  $k > 2\exp(-3\eps)$, it will suffice to have
\begin{align*}
\beta &> n \cdot \left(1-\exp(-5\eps)\right)^T \\
\ln \frac{n}{\beta} &< T\ln\left(\frac{1}{1-\exp(-5\eps)}\right)\\
&= T\ln\left(1+\frac{1}{\exp(5\eps)-1}\right) \stepcounter{equation} \label{eq:pre-lower-bound-T} \tag{\theequation}
\end{align*}

The following is fairly trivial to prove: if $0 < \tau < 1$, then $1+\tau > \exp(\tau / 2)$. Here, $\tau :=(\exp(5\eps)-1)^{-1}$. We are ensured that $\tau \in (0, 1)$ because $\eps > \ln(2)/5$.  Therefore, the following bound on $\ln (n /\beta)$ is tighter than \eqref{eq:pre-lower-bound-T}:
\begin{align*}
\ln \frac{n}{\beta} &< T \ln \left( \exp\left(\half\cdot \frac{1}{\exp(5\eps)-1} \right) \right)\\
&= T\cdot \half\cdot \frac{1}{\exp(5\eps)-1}\\
T &> \ln \frac{n}{\beta} \cdot 2(\exp(5\eps)-1) \stepcounter{equation} \tag{\theequation} \label{eq:lower-bound-T}
\end{align*}

We also want the second term of \eqref{eq:three-beta-terms} to be bounded by $\beta$.
\begin{align*}
\beta &> \frac{2nT\delta\exp(\eps)}{1-\exp(-\eps)}\\
T &< \beta \frac{1-\exp(-\eps)}{2n\delta\exp(\eps)} \stepcounter{equation} \tag{\theequation} \label{eq:upper-bound-T}
\end{align*}

Both \eqref{eq:lower-bound-T} and \eqref{eq:upper-bound-T} need be true for the same value of $T$. Hence,
\begin{align*}
\ln \frac{n}{\beta} \cdot 2(\exp(5\eps)-1) &< \beta \frac{1-\exp(-\eps)}{2n\delta\exp(\eps)}\\
\delta &< \beta \frac{1-\exp(-\eps)}{2n\exp(\eps)}\cdot \frac{1}{\ln \frac{n}{\beta} \cdot 2(\exp(5\eps)-1)}\\
&= \frac{\beta}{4n\ln(n/\beta)}\cdot \frac{\exp(\eps)-1}{\exp(2\eps)(\exp(5\eps)-1)} \stepcounter{equation} \label{eq:upper-bound-delta-loose} \tag{\theequation}
\end{align*}

Because $\eps > 2/3$, one can show that
\[
\frac{\exp(\eps)-1}{\exp(2\eps)(\exp(5\eps)-1)} > \half \exp(-6\eps)
\]
which means that any $\delta$ satisfying \eqref{eq:upper-bound-delta} satisfies \eqref{eq:upper-bound-delta-loose}.

The final term in \eqref{eq:three-beta-terms} is $\frac{4n\delta\exp(\eps)}{1-\exp(-\eps)}$; for this to be bounded by $\beta$, it will suffice for
\[
\delta < \beta\cdot \frac{0.1}{n\exp(\eps)} \tag{$\eps > 2/3$}
\]
This constraint on $\delta$ is not as tight as \eqref{eq:upper-bound-delta} whenever $n > e$.

Because we have shown all three terms in \eqref{eq:three-beta-terms} are bounded by $\beta$, this concludes the proof.
\end{proof}

\medskip

\paragraph{Setting parameters $k,T$} We now provide parameter values to ensure our transformation is $(\alpha,4\beta)$ accurate, thereby proving Claim \ref{clm:transform-is-accurate}.
Suppose the parameters $k,T$ are assigned as follows
\begin{align*}
k &\gets 2\exp(-2.5\eps)\\
T &\gets \big\lceil \ln \frac{n}{\beta} \cdot 2(\exp(5\eps)-1) \big\rceil
\end{align*}

Because $k \in (0, 2\exp(-2\eps))$ and $T\in N$, Claim \ref{clm:transform-is-accurate-2} implies that the protocol $P_{n,k,T}=(R_{n,k,T},A_n)$ is $(\alpha,\beta_{k,T})$ accurate, where $\beta_{k,T}$ is defined in \eqref{eq:bkt}. Claim \ref{clm:bounding-the-gap} implies that there is an integer value of $T$ where $\beta_{k,T} \leq 4\beta$; from \eqref{eq:lower-bound-T}, $T$ is assigned such a value. Hence, $P_{n,k,T}$ is $(\alpha,4\beta)$ accurate.

\section{Concentration Inequalities}
In this appendix, we formally state the three concentration inequalities used in this paper:

\begin{thm}[Chernoff bound]
\label{thm:chernoff}
If $\mathbf{x_1,\dots, x_n}$ are independent $\zo$-valued random variables, each with mean $\mu$, then, for every $\beta > 0$,
\[
\pr{}{\mu n - \sum \mathbf{x_i} < \sqrt{2 \mu  n \log\tfrac{1}{\beta}}} \geq 1-\beta, \textrm{ and}
\]
\[
\pr{}{\sum \mathbf{x_i} - \mu n < \sqrt{3 \mu  n \log\tfrac{1}{\beta}}} \geq 1-\beta
\]
\end{thm}

\begin{thm}[Hoeffding's inequality]
\label{thm:hoeffding}
If $\mathbf{x_1,\dots, x_n}$ are independent random variables, each with mean $\mu$ and bounded in $(a,b)$, then, for every $\beta > 0$,
\[
\pr{}{\left| \sum \mathbf{x_i} - \mu  n \right| < (b-a)\sqrt{\tfrac12 n \log \tfrac{2}{\beta}}} > 1-\beta
\]
\end{thm}

\begin{thm}[Bernstein's inequality]
\label{thm:bernstein}
If $\mathbf{x_1,\dots, x_n}$ are independent random variables, each with mean $0$, variance $\sigma^2 > \frac{4}{9n}\log\frac{2}{\beta}$, and bounded in $[-1,1]$, then, for every $\beta > 0$,
\[
\pr{}{\left| \sum \mathbf{x_i} \right| < 2 \sigma\sqrt{n\log \tfrac{2}{\beta}}} > 1-\beta
\]
\end{thm}
\fi

\end{document}